\documentclass[a4paper,reqno]{amsart}

\usepackage[british]{babel}

\numberwithin{equation}{section}

\usepackage[left=3.5cm,right=3.5cm,top=3.5cm,bottom=3.5cm,footskip=1.5cm,headsep=1cm]{geometry}
\usepackage[%
pdfauthor={Habib Ammari, Silvio Barandun, Bryn Davies,
Erik Orvehed Hiltunen, Thea Kosche, and Ping Liu},%
pdftitle={Interface modes in Hermitian Su-Schrieffer-Heeger chains of resonators},
pdfsubject={Interface modes},
pdfkeywords={Hermitian systems, subwavelength resonators, interface modes, topological protection, localisation},
]{hyperref}
\usepackage[T1]{fontenc} 
\usepackage{lmodern} 
\rmfamily 
\DeclareFontShape{T1}{lmr}{b}{sc}{<->ssub*cmr/bx/sc}{}
\DeclareFontShape{T1}{lmr}{bx}{sc}{<->ssub*cmr/bx/sc}{}
\usepackage{lipsum}
\usepackage{mathtools}
\usepackage{amsmath}
\usepackage{amssymb}
\usepackage{tikz}
\usetikzlibrary{matrix,positioning,decorations.pathreplacing,calc}
\usetikzlibrary{
decorations.text,%
decorations.markings,%
shadows}
\usepackage{adjustbox}
\usepackage{graphicx} 

\usepackage{caption}
\usepackage{subcaption}

\usepackage{dsfont} 
\usepackage[format=hang]{caption}

\usepackage[normalem]{ulem}

\usepackage{bm}

\usepackage[
style=numeric,
sorting=nty,
maxnames=99,
maxalphanames=5,
natbib=true,
sortcites]{biblatex}

\DeclareNameAlias{default}{family-given}

\graphicspath{{figures/}}

\AtEveryBibitem{
 \clearfield{url}
 \clearfield{issn}
 \clearfield{isbn}
 \clearfield{urldate}
 
 \ifentrytype{book}{
  \clearfield{pages}}{%
 }
}

\usepackage{xargs}                      
\usepackage[prependcaption,textsize=tiny,textwidth=3cm]{todonotes}
\newcommandx{\unsure}[2][1=]{\todo[linecolor=red,backgroundcolor=red!25,bordercolor=red,#1]{#2}}
\newcommandx{\change}[2][1=]{\todo[linecolor=blue,backgroundcolor=blue!25,bordercolor=blue,#1]{#2}}
\newcommandx{\info}[2][1=]{\todo[linecolor=OliveGreen,backgroundcolor=OliveGreen!25,bordercolor=OliveGreen,#1]{#2}}
\newcommandx{\improvement}[2][1=]{\todo[linecolor=black,backgroundcolor=black!25,bordercolor=black,#1]{#2}}
\newcommandx{\thiswillnotshow}[2][1=]{\todo[disable,#1]{#2}}
\usepackage{amsthm}
\usepackage[noabbrev,capitalize]{cleveref}
\crefname{proposition}{Proposition}{Propositions}
\crefname{equation}{}{}

\newtheorem{theorem}{Theorem}[section]
\newtheorem{lemma}[theorem]{Lemma}
\newtheorem{proposition}[theorem]{Proposition}

\theoremstyle{definition}
\newtheorem{definition}[theorem]{Definition}

\newtheorem{remark}[theorem]{Remark}

\crefname{assumption}{Assumption}{Assumptions}
\crefname{definition}{Definition}{Definitions}
\crefname{corollary}{Corollary}{Corollaries}
\crefname{enumi}{item}{items}

\usepackage{fancyhdr}

\fancypagestyle{plain}{%
\fancyhead[C]{}
\fancyfoot[C]{}

}
\fancypagestyle{nosection}{%
\fancyhead[CE]{}
\fancyhead[CO]{}
\fancyfoot[CE,CO]{\thepage}
}

\DeclareMathOperator{\N}{\mathbb{N}}

\DeclareMathOperator{\R}{\mathbb{R}}
\DeclareMathOperator{\C}{\mathbb{C}}

\renewcommand{\phi}{\varphi}

\renewcommand{\tilde}{\widetilde}
\renewcommand{\hat}{\widehat}

\newcommand{\inv}{^{-1}}

\renewcommand{\qedsymbol}{$\blacksquare$}

\usepackage{fancyhdr}

\fancypagestyle{plain}{%
\fancyhead[C]{}
\fancyfoot[C]{}

}
\fancypagestyle{nosection}{%
\fancyhead[CE]{}
\fancyhead[CO]{}
\fancyfoot[CE,CO]{\thepage}
}

\DeclareMathOperator{\diag}{diag}
\DeclareMathOperator{\BO}{\mathcal{O}}

\newcommand{\pri}{^\prime}

\renewcommand{\epsilon}{\varepsilon}
\DeclareMathOperator{\dd}{d\!}

\renewcommand{\tilde}{\widetilde}
\renewcommand{\hat}{\widehat}

\DeclareMathOperator{\iL}{{\mathsf{L}}}
\DeclareMathOperator{\iR}{{\mathsf{R}}}
\DeclareMathOperator{\iLR}{{\mathsf{L},\mathsf{R}}}

\usepackage{tcolorbox}
\usetikzlibrary{tikzmark,calc}

\DeclareMathOperator{\antidiag}{antidiag}
\DeclareMathOperator{\sspan}{span}

\usepackage{enumerate}
\usepackage{upgreek}
\usepackage{nicematrix}

\newcommand{\leqs}{\leqslant}
\newcommand{\geqs}{\geqslant}
\renewcommand{\leq}{\leqs}
\renewcommand{\geq}{\geqs}

\begin{document}

\title{Exponentially localised interface eigenmodes in finite chains of resonators}

  \author[H. Ammari]{Habib Ammari}
 \address{\parbox{\linewidth}{Habib Ammari\\
  ETH Z\"urich, Department of Mathematics, Rämistrasse 101, 8092 Z\"urich, Switzerland}}
 \email{habib.ammari@math.ethz.ch}
 \thanks{}

 \author[S. Barandun]{Silvio Barandun}
  \address{\parbox{\linewidth}{Silvio Barandun\\
 ETH Z\"urich, Department of Mathematics, Rämistrasse 101, 8092 Z\"urich, Switzerland}}
  \email{silvio.barandun@sam.math.ethz.ch}

  \author[B. Davies]{Bryn Davies}
 \address{\parbox{\linewidth}{Bryn Davies\\
Department of Mathematics, Imperial College London, 180 Queen's Gate, London SW7~2AZ, UK}}
\email{bryn.davies@imperial.ac.uk}

 \author[E.O. Hiltunen]{Erik Orvehed Hiltunen}
\address{\parbox{\linewidth}{Erik Orvehed Hiltunen\\
Department of Mathematics, Yale University, 10 Hillhouse Ave,
New Haven, CT~06511, USA}}
\email{erik.hiltunen@yale.edu}

 \author[T. Kosche]{Thea Kosche}
  \address{\parbox{\linewidth}{Thea Kosche\\
 ETH Z\"urich, Department of Mathematics, Rämistrasse 101, 8092 Z\"urich, Switzerland}}
  \email{thea.kosche@sam.math.ethz.ch}

 \author[P. Liu]{Ping Liu}
  \address{\parbox{\linewidth}{Ping Liu\\
 ETH Z\"urich, Department of Mathematics, Rämistrasse 101, 8092 Z\"urich, Switzerland}}
  \email{ping.liu@sam.math.ethz.ch}

\maketitle

\begin{abstract}
This paper studies wave localisation in chains of finitely many resonators. There is an extensive theory predicting the existence of localised modes induced by defects in infinitely periodic systems. This work extends these principles to finite-sized systems. We consider finite systems of subwavelength resonators arranged in dimers that have a geometric defect in the structure. This is a classical wave analogue of the Su-Schrieffer-Heeger model. We prove the existence of a spectral gap for defectless finite dimer structures and find a direct relationship between eigenvalues being within the spectral gap and the localisation of their associated eigenmode. Then we show the existence and uniqueness of an eigenvalue in the gap in  the defect structure, proving the existence of a unique localised interface mode. 
To the best of our knowledge, our method, based on Chebyshev polynomials, is the first  to characterise quantitatively the localised interface modes in systems of finitely many resonators.
\end{abstract}

\date{}

\bigskip

  \par

\bigskip


\section{Introduction}

Wave localisation at subwavelength scales has many important applications in nanophotonics and nanophononics \cite{review1,review2,review3}. Here, \emph{subwavelength} means that the incident wavelengths are much larger than the size of the building blocks of the structure. When these relatively small building blocks are locally resonant (which, in the case studied here, will be due to large material contrasts) this allows for waves to be localised at subwavelength scales, thereby beating traditional diffraction limits \cite{sima2023,anderson2022,essentiel2023}.  This principle has unlocked a wealth of novel nanotechnologies \cite{sheng,cummer}.

In this paper, we consider wave localisation at the interface between two systems of \emph{finitely many} subwavelength resonators. These resonators are arranged in pairs or dimers, such that the model we consider shares many of the features of the Su-Schrieffer-Heeger (SSH) model in quantum mechanics \cite{original_ssh}. We prove the existence of exponentially localised interface eigenmodes in this finite structure. These interface modes have been subject to numerous studies in the setting of infinite structures. A particular focus has been put on studying the topological properties of infinite periodic structures and then introducing carefully designed interfaces so as to create so-called \emph{topologically protected} eigenmodes \cite{ssh3d}. These modes have significant implications for applications since they are expected to be robust with respect to imperfections in the design. These concepts have been widely studied in a variety of settings, most notably in quantum mechanics for the Schrödinger operator \cite{fefferman1,fefferman2} and more recently, for related continuous models of classical wave systems in \cite{hai2022,hai2023,haiarxiv,bryn,guo2023}. In finite-sized systems, these eigenmodes have been observed both experimentally and numerically; see, for instance, \cite{ammari.ea2023Edge,ssh3d,jasa,jasa2} and references therein.

The present work considers the far less-explored but more realistic physical setting of interface modes in finite dimer structures. As far as we know, it is the first work to deal with the existence of interface modes in finite structures. It provides a one-to-one correspondence between the position of the eigenfrequency in the spectrum of the corresponding infinite periodic structure (\emph{i.e.}, in the asymptotic spectral bulk, its boundary or in the asymptotic spectral gap; see \Cref{def:spectralgap}) and the behaviour (localised versus delocalised) of the corresponding eigenmode. Our results hold for any finite and large enough system of dimer structure with defect satisfying a mild condition on the size of the resonators. Furthermore, we show that the eigenfrequencies lying in the gap converge exponentially as the size of the structure goes to infinity and provide an explicit and simple formula for the limit. We also show that the Hermitian nature of the system together with the position of the interface eigenfrequency in the gap yields a very strong stability property when the geometry of the system is perturbed.

Our approach is based on the capacitance matrix, which in general is a powerful tool for characterising the subwavelength eigenfrequencies of a system of high-contrast resonators \cite{ammari.davies.ea2021Functional, feppon.cheng.ea2023Subwavelength,ammari.ea2023Edge}. In essence, the capacitance formulation provides a discrete approximation to the continuous spectral problem of the differential model, valid in the high-contrast asymptotic limit. This approximation is based solely on first principles and provides a natural starting point for both theoretical analysis and numerical simulation of wave localisation in the subwavelength regime \cite{essentiel2023,ammari2023perturbed,anderson2022}. In the case of 
our one-dimensional, finite system of dimer resonators with a geometric defect, the capacitance matrix is a perturbed tridiagonal block $2$-Toeplitz matrix; see \eqref{eq: strucutre capacitance matrix}. Based on some properties of the Chebyshev polynomials, we prove existence and uniqueness of an eigenfrequency in the gap for the finite dimer structures with a geometric defect and show exponential localisation of the corresponding eigenmode. Our proof does not rely on any perturbation argument neither on any \emph{a prior} assumption on the band gaps of the two dimer systems in the structure.  
    
The paper is organised as follows. In \Cref{sect2}, we introduce the capacitance matrix approximation of a finite, dimer system with a geometric defect in order to approximate its eigenfrequencies and associated eigenmodes. In \Cref{sect3}, we first write the characteristic polynomial of the capacitance matrix associated to a defectless, finite dimer system in terms of Chebyshev polynomials. Then, we characterise the structure of the eigenvectors of the capacitance matrix associated with the perturbed dimer structure. \Cref{sect4} is devoted to the derivation of a direct relationship between an eigenvalue being within the spectral gap and the localisation of its corresponding eigenvector. In \Cref{sect5}, we show the existence and uniqueness of an eigenvalue of the capacitance matrix lying in the asymptotic spectral gap and consequently the existence of a unique localised interface eigenvector. Furthermore, we show that this eigenvalue converges exponentially fast to a value in the gap as the system size increases. In \Cref{sect6}, we analyse the robustness of the interface localisation with respect to imperfections in the structure design. The paper ends with some concluding remarks in \Cref{sect7}. In \Cref{appA}, we recall some basic definitions and results on pseudo-spectra of normal matrices. \Cref{appC} is devoted to the discussion of the topological origin of the robustness of the interface eigenmodes. 


\section{One-dimensional subwavelength resonator systems} \label{sect2}
We consider a one-dimensional chain of $N$ disjoint identical high-contrast resonators $D_i\coloneqq (x_i^{\iL},x_i^{\iR})$, where $(x_i^{\iLR})_{1\<i\<N} \subset \R$ are the $2N$ boundaries satisfying $x_i^{\iL} < x_i^{\iR} <  x_{i+1}^{\iL}$ for any $1\leq i \leq N-1$. We fix the coordinates such that $x_1^{\iL}=0$. We also denote by  $\ell_i = x_i^{\iR} - x_i^{\iL}$ the length of the each of the resonators, and by $s_i= x_{i+1}^{\iL} -x_i^{\iR}$ the spacing between the $i$-th and $(i+1)$-th inclusions. The system is illustrated in \cref{fig:setting}. 

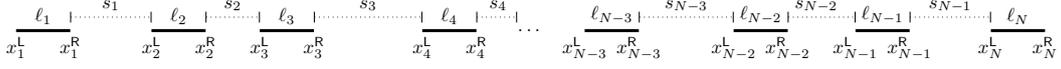
\begin{figure}[htb]
    \centering
    \begin{adjustbox}{width=\textwidth}
    \begin{tikzpicture}
        \coordinate (x1l) at (1,0);
        \path (x1l) +(1,0) coordinate (x1r);
        \path (x1r) +(0.75,0.7) coordinate (s1);
        \path (x1r) +(1.5,0) coordinate (x2l);
        \path (x2l) +(1,0) coordinate (x2r);
        \path (x2r) +(0.5,0.7) coordinate (s2);
        \path (x2r) +(1,0) coordinate (x3l);
        \path (x3l) +(1,0) coordinate (x3r);
        \path (x3r) +(1,0.7) coordinate (s3);
        \path (x3r) +(2,0) coordinate (x4l);
        \path (x4l) +(1,0) coordinate (x4r);
        \path (x4r) +(0.4,0.7) coordinate (s4);
        \path (x4r) +(1,0) coordinate (dots);
        \path (dots) +(1,0) coordinate (x5l);
        \path (x5l) +(1,0) coordinate (x5r);
        \path (x5r) +(1.75,0) coordinate (x6l);
        \path (x5r) +(0.875,0.7) coordinate (s5);
        \path (x6l) +(1,0) coordinate (x6r);
        \path (x6r) +(1.25,0) coordinate (x7l);
        \path (x6r) +(0.525,0.7) coordinate (s6);
        \path (x7l) +(1,0) coordinate (x7r);
        \path (x7r) +(1.5,0) coordinate (x8l);
        \path (x7r) +(0.75,0.7) coordinate (s7);
        \path (x8l) +(1,0) coordinate (x8r);
        \draw[ultra thick] (x1l) -- (x1r);
        \node[anchor=north] (label1) at (x1l) {$x_1^{\iL}$};
        \node[anchor=north] (label1) at (x1r) {$x_1^{\iR}$};
        \node[anchor=south] (label1) at ($(x1l)!0.5!(x1r)$) {$\ell_1$};
        \draw[dotted,|-|] ($(x1r)+(0,0.25)$) -- ($(x2l)+(0,0.25)$);
        \draw[ultra thick] (x2l) -- (x2r);
        \node[anchor=north] (label1) at (x2l) {$x_2^{\iL}$};
        \node[anchor=north] (label1) at (x2r) {$x_2^{\iR}$};
        \node[anchor=south] (label1) at ($(x2l)!0.5!(x2r)$) {$\ell_2$};
        \draw[dotted,|-|] ($(x2r)+(0,0.25)$) -- ($(x3l)+(0,0.25)$);
        \draw[ultra thick] (x3l) -- (x3r);
        \node[anchor=north] (label1) at (x3l) {$x_3^{\iL}$};
        \node[anchor=north] (label1) at (x3r) {$x_3^{\iR}$};
        \node[anchor=south] (label1) at ($(x3l)!0.5!(x3r)$) {$\ell_3$};
        \draw[dotted,|-|] ($(x3r)+(0,0.25)$) -- ($(x4l)+(0,0.25)$);
        \node (dots) at (dots) {\dots};
        \draw[ultra thick] (x4l) -- (x4r);
        \node[anchor=north] (label1) at (x4l) {$x_4^{\iL}$};
        \node[anchor=north] (label1) at (x4r) {$x_4^{\iR}$};
        \node[anchor=south] (label1) at ($(x4l)!0.5!(x4r)$) {$\ell_4$};
        \draw[dotted,|-|] ($(x4r)+(0,0.25)$) -- ($(dots)+(-.25,0.25)$);
        \draw[ultra thick] (x5l) -- (x5r);
        \node[anchor=north] (label1) at (x5l) {$x_{N-3}^{\iL}$};
        \node[anchor=north] (label1) at (x5r) {$x_{N-3}^{\iR}$};
        \node[anchor=south] (label1) at ($(x5l)!0.5!(x5r)$) {$\ell_{N-3}$};
        \draw[dotted,|-|] ($(x5r)+(0,0.25)$) -- ($(x6l)+(0,0.25)$);
        \draw[ultra thick] (x6l) -- (x6r);
        \node[anchor=north] (label1) at (x6l) {$x_{N-2}^{\iL}$};
        \node[anchor=north] (label1) at (x6r) {$x_{N-2}^{\iR}$};
        \node[anchor=south] (label1) at ($(x6l)!0.5!(x6r)$) {$\ell_{N-2}$};
        \draw[dotted,|-|] ($(x6r)+(0,0.25)$) -- ($(x7l)+(0,0.25)$);
        \draw[ultra thick] (x7l) -- (x7r);
        \node[anchor=north] (label1) at (x7l) {$x_{N-1}^{\iL}$};
        \node[anchor=north] (label1) at (x7r) {$x_{N-1}^{\iR}$};
        \node[anchor=south] (label1) at ($(x7l)!0.5!(x7r)$) {$\ell_{N-1}$};
        \draw[dotted,|-|] ($(x7r)+(0,0.25)$) -- ($(x8l)+(0,0.25)$);
        \draw[ultra thick] (x8l) -- (x8r);
        \node[anchor=north] (label1) at (x8l) {$x_{N}^{\iL}$};
        \node[anchor=north] (label1) at (x8r) {$x_{N}^{\iR}$};
        \node[anchor=south] (label1) at ($(x8l)!0.5!(x8r)$) {$\ell_N$};
        \node[anchor=north] (label1) at (s1) {$s_1$};
        \node[anchor=north] (label1) at (s2) {$s_2$};
        \node[anchor=north] (label1) at (s3) {$s_3$};
        \node[anchor=north] (label1) at (s4) {$s_4$};
        \node[anchor=north] (label1) at (s5) {$s_{N-3}$};
        \node[anchor=north] (label1) at (s6) {$s_{N-2}$};
        \node[anchor=north] (label1) at (s7) {$s_{N-1}$};
    \end{tikzpicture}
    \end{adjustbox}
    \caption{A chain of $N$ resonators, with lengths
    $(\ell_i)_{1\leq i\leq N}$ and spacings $(s_{i})_{1\leq i\leq N-1}$ (which will be chosen to alternate between two distinct values, as depicted).}
    \label{fig:setting}
\end{figure}

We use 
\begin{align*}
   D\coloneqq \bigcup_{i=1}^N(x_i^{\iL},x_i^{\iR})
\end{align*}
to denote the set of subwavelength resonators. In this paper, we only consider systems of identically sized resonators, that is
\begin{align*}
    \ell_i = \ell \in \R_{>0}\text{ for all } 1\leq i\leq N.
\end{align*}
This will simplify the formulas in our subsequent analysis and is sufficient to observe the physical phenomena we are interested in.

In this work, we consider the one-dimensional wave equation propagating in a heterogeneous medium with space-dependent material parameters:
\begin{align}
    \frac{\omega^{2}}{\kappa(x)}u(x) +\frac{\dd}{\dd x}\left( \frac{1}{\rho(x)}\frac{\dd}{\dd
    x}  u(x)\right) =0,\qquad x \in\R.
    \label{eq: gen Strum-Liouville}
\end{align}
The material parameters $\kappa(x)$ and $\rho(x)$ are piecewise constant in the interior and exterior of the resonators
\begin{align*}
    \kappa(x)=
    \begin{dcases}
        \kappa_b, & x\in D,\\
        \kappa, &  x\in\R\setminus D,
    \end{dcases}\quad\text{and}\quad
    \rho(x)=
    \begin{dcases}
        \rho_b, & x\in D,\\
        \rho, &  x\in\R\setminus D,
    \end{dcases}
\end{align*}
where the constants $\rho_b, \rho, \kappa, \kappa_b \in \R_{>0}$. The wave speeds inside the set $D$ of resonators  and inside the background medium $\R\setminus D$, are denoted respectively by $v_b$ and $v$, the wave numbers respectively by $k_b$ and $k$, and the contrast between the densities of the resonators and the background medium by $\delta$:
\begin{align}
    v_b:=\sqrt{\frac{\kappa_b}{\rho_b}}, \qquad v:=\sqrt{\frac{\kappa}{\rho}},\qquad
    k_b:=\frac{\omega}{v_b},\qquad k:=\frac{\omega}{v},\qquad
    \delta:=\frac{\rho_b}{\rho}.
\end{align}

For these step-wise defined material parameters, the wave problem determined by \eqref{eq: gen Strum-Liouville} reduces to the following system of coupled one-dimensional Helmholtz equations:
\begin{align}
    \label{eq: system of coupled equations}
    \begin{dcases}
        \frac{\dd{^2}}{\dd x^2}u(x)+ \frac{\omega^2}{v^2}u(x) = 0, & x\in \R \setminus D ,\\
        \frac{\dd{^2}}{\dd x^2}u(x)+ \frac{\omega^2}{v_b^2}u(x) = 0, & x\in D ,\\
        u\vert_{\iR}(x^{\iLR}_{{i}}) - u\vert_{\iL}(x^{\iLR}_{{i}}) = 0, &  1\leq i\leq N ,\\
        \left.\frac{\dd u}{\dd x}\right\vert_{\iR}(x^{\iL}_{{i}}) - \delta\left.\frac{\dd u}{\dd x}\right\vert_{\iL}(x^{\iL}_{{i}}) = 0, & 1\leq i\leq N ,\\
        \delta\left.\frac{\dd u}{\dd x}\right\vert_{\iR}(x^{\iR}_{{i}}) - \left.\frac{\dd u}{\dd x}\right\vert_{\iL}(x^{\iR}_{{i}}) = 0, & 1\leq i\leq N,\\
        \big(\frac{\dd}{\dd |x|} - \mathrm{i} k \big) u = 0 & \text{for } x \in (-\infty, x^{\iL}_{1}) \cup (x^{\iR}_{N}, +\infty),  \\
    \end{dcases}
\end{align}
where for a one-dimensional function $w$ we denote by
\begin{align*}
    w\vert_{\iL}(x) \coloneqq \lim_{\substack{s\to 0\\ s>0}}w(x-s) \quad \mbox{and} \quad  w\vert_{\iR}(x) \coloneqq \lim_{\substack{s\to 0\\ s>0}}w(x+s)
\end{align*}
if the limits exist.

We are interested in the high-contrast regime characterised by the contrast parameter $\delta$ being small. In this case, there exist resonant modes of the system at subwavelength frequencies, for which the size of the resonators is smaller than the wavelength in the background medium.
\begin{definition}
    A resonance $\omega(\delta)\in\C$ for which there exists a non-trivial solution to \eqref{eq: system of coupled equations} is called \emph{subwavelength} in the high-contrast regime if
    \begin{align*}
        \omega(\delta)\to 0 \quad \text{as}\quad \delta \to 0.
    \end{align*}
\end{definition}
One consequence of this asymptotic ansatz is that it lends itself to characterisation using asymptotic analysis \cite{ammari.davies.ea2021Functional}. This limit recovers subwavelength resonances, while keeping the size of the resonators fixed. 

In \cite{feppon.cheng.ea2023Subwavelength}, an asymptotic analysis in the subwavelength limit was performed on the system of one-dimensional subwavelength resonators considered here. It was shown that the leading-order behavior of the resonances is given by the eigenpairs of the \emph{capacitance matrix}:
\begin{align}\label{eq: general capacitance matrix}
    \mathcal{C}\coloneqq \begin{pmatrix}
        \frac{1}{s_1} & -\frac{1}{s_1}\\
        -\frac{1}{s_1} & \frac{1}{s_1} + \frac{1}{s_2} & -\frac{1}{s_2}\\
         & -\frac{1}{s_2} & \frac{1}{s_2} + \frac{1}{s_3} & -\frac{1}{s_3}\\
         && \ddots & \ddots & \ddots \\
         &&& -\frac{1}{s_{N-2}} & \frac{1}{s_{N-2}} + \frac{1}{s_{N-1}} & -\frac{1}{s_{N-1}}\\
         &&&& -\frac{1}{s_{N-1}} & \frac{1}{s_{N-1}}
    \end{pmatrix}.
\end{align}
This is a modified version of the conventional capacitance matrix that is often used to characterise many-body low-frequency resonance problems. We summaries the findings of \cite{feppon.cheng.ea2023Subwavelength} here below.

\begin{proposition}\label{prop: reduction to capacitance matrix}
Consider a system of $N$ subwavelength resonators with size $\ell$ and spacings $s_i$ for $1\leq i \leq N-1$. Assume that the eigenvalues of $\mathcal{C}$ are simple. Then, the $N$ subwavelength resonant frequencies $\omega_i$ of \eqref{eq: system of coupled equations} satisfy to the first order
    \begin{align*}
        \omega_i =  v_b\sqrt{\delta\lambda_i} + \BO(\delta),
    \end{align*}
    where $(\lambda_i)_{1\leq i\leq N}$ are the eigenvalues of the eigenvalue problem
\begin{equation}
\label{eq:eigevalue problem capacitance matrix}
\mathcal{C} \bm a_i = \lambda_i \ell \bm a_i\qquad 1\leq i\leq N.
\end{equation}
Furthermore, let $u_i(x)$ be a subwavelength eigenmode corresponding to $\omega_i$ and let $\bm a_i$ be the corresponding eigenvector of $\mathcal{C}$. Then
\begin{align*}
    u_i(x) = \sum_{j=1}^N\bm a_i^{(j)}V_j(x) + \BO(\delta)
\end{align*}
where $\bm a^{(j)}_i$ denotes the $j$-th entry of the eigenvector $\bm a_i$ and $V_j(x)$ is piecewise linear, supported in $(x_{j-1}^{\iR},x_{j+1}^{\iL})$ and $V_j(x)=1$ for $x\in(x_{j}^{\iL},x_{j}^{\iR})$.
\end{proposition}
\cref{prop: reduction to capacitance matrix} states that eigenvalues and eigenvectors of the capacitance matrix stand in a one-to-one correspondence with subwavelength resonant frequencies and modes of the associated system.
\subsection{Dimer chains with a defect}

Systems of repeated dimers (that is $s_i=s_{i-2}$ for $3\leq i\leq N$) are of particular interest as the corresponding infinite structure can be studied with Floquet--Bloch band theory \cite{ammari.ea2023Edge}; when the two repeating separation distances are distinct, this  provides a non-trivial example of a band gap between the subwavelength spectral bands. Here, we consider systems of dimers with a defect in the geometric structure, so that at some point the repeating pattern of alternating separation distances is broken. This is inspired by the famous Su-Schrieffer-Heeger (SSH) model from quantum settings and is the canonical example of a topologically protected interface mode \cite{original_ssh}. It is a system of $N=4m+1$ for $m\in\N$ resonators such that
\begin{align*}
    s_i &= s_{i-2} \quad \text{for } 3\leq i \leq 2m,\\
    s_i &= s_{i+2} \quad \text{for } 2m+1\leq i \leq 4m-2,
\end{align*}
where we typically assume $s_1=s_{2m+2}$ and $s_2=s_{2m+1}$ so that the system is symmetric with respect to the center of the $(2m+1)$-th resonator with spacings $s_1$ and $s_2$. A sketch of this system with a defect is shown in \cref{fig: geometrical defect}.

\begin{figure}[h]
    \centering
    \begin{adjustbox}{width=\textwidth}
        \begin{tikzpicture}
        \draw[-,thick,dotted] (-.5,-1) -- (-.5,2);
        \draw[|-|,dashed] (0,1) -- (1,1);
        \node[above] at (0.5,1) {$s_2$};
        \draw[ultra thick] (1,0) -- (2,0);
        \node[below] at (1.5,0) {$D_{2m+2}$};
        \draw[-,dotted] (1,0) -- (1,1);
        
        \draw[|-|,dashed] (2,1) -- (3,1);
        \node[above] at (2.5,1) {$s_1$};
        \draw[ultra thick] (3,0) -- (4,0);
        \node[below] at (3.5,0) {$D_{2m+3}$};
        \draw[-,dotted] (2,0) -- (2,1);
        \draw[-,dotted] (3,0) -- (3,1);
        \draw[-,dotted] (4,0) -- (4,1);
        
        \begin{scope}[shift={(+4,0)}]
        \draw[|-|,dashed] (0,1) -- (1,1);
        \node[above] at (0.5,1) {$s_2$};
        \draw[ultra thick] (1,0) -- (2,0);
        \node[below] at (1.5,0) {$D_{2m+4}$};
        \draw[-,dotted] (1,0) -- (1,1);
        \node at (2.5,.5) {\dots};
        \end{scope}

        \begin{scope}[shift={(+6,0)}]
        \draw[ultra thick] (1,0) -- (2,0);
        \node[below] at (1.5,0) {$D_{4m-2}$};
        
        \draw[|-|,dashed] (2,1) -- (3,1);
        \node[above] at (2.5,1) {$s_2$};
        \draw[ultra thick] (3,0) -- (4,0);
        \node[below] at (3.5,0) {$D_{4m}$};
        \draw[-,dotted] (2,0) -- (2,1);
        \draw[-,dotted] (3,0) -- (3,1);
        \draw[-,dotted] (4,0) -- (4,1);
\begin{scope}[shift={(+2,0)}]
        \draw[|-|,dashed] (2,1) -- (3,1);
        \node[above] at (2.5,1) {$s_1$};
        \draw[ultra thick] (3,0) -- (4,0);
        \node[below] at (3.5,0) {$D_{4m+1}$};
        \draw[-,dotted] (2,0) -- (2,1);
        \draw[-,dotted] (3,0) -- (3,1);
        \end{scope}
        \end{scope}
        
\begin{scope}[shift={(-4,0)}]

        \draw[|-|,dashed] (0,1) -- (1,1);
        \node[above] at (0.5,1) {$s_1$};
        \draw[ultra thick] (1,0) -- (2,0);
        \node[below] at (1.5,0) {$D_{2m}$};
        \draw[-,dotted] (1,0) -- (1,1);
        
        \draw[|-|,dashed] (2,1) -- (3,1);
        \node[above] at (2.5,1) {$s_2$};
        \draw[ultra thick] (3,0) -- (4,0);
        \node[below,fill=white] at (3.5,0) {$D_{2m+1}$};
        \draw[-,dotted] (2,0) -- (2,1);
        \draw[-,dotted] (3,0) -- (3,1);
        \draw[-,dotted] (4,0) -- (4,1);
        \end{scope}
        
        \begin{scope}[shift={(-8,0)}]
        \node at (.5,.5) {\dots};
        \draw[ultra thick] (1,0) -- (2,0);
        \node[below] at (1.5,0) {$D_{2m-2}$};
        
        \draw[|-|,dashed] (2,1) -- (3,1);
        \node[above] at (2.5,1) {$s_2$};
        \draw[ultra thick] (3,0) -- (4,0);
        \node[below] at (3.5,0) {$D_{2m-3}$};
        \draw[-,dotted] (2,0) -- (2,1);
        \draw[-,dotted] (3,0) -- (3,1);
        \draw[-,dotted] (4,0) -- (4,1);
        \end{scope}

        \begin{scope}[shift={(-14,0)}]
        \draw[ultra thick] (1,0) -- (2,0);
        \node[below] at (1.5,0) {$D_{1}$};
        
        \draw[|-|,dashed] (2,1) -- (3,1);
        \node[above] at (2.5,1) {$s_1$};
        \draw[ultra thick] (3,0) -- (4,0);
        \node[below] at (3.5,0) {$D_{2}$};
        \draw[-,dotted] (2,0) -- (2,1);
        \draw[-,dotted] (3,0) -- (3,1);
        \draw[-,dotted] (4,0) -- (4,1);
\begin{scope}[shift={(+2,0)}]
        \draw[|-|,dashed] (2,1) -- (3,1);
        \node[above] at (2.5,1) {$s_2$};
        \draw[ultra thick] (3,0) -- (4,0);
        \node[below] at (3.5,0) {$D_{3}$};
        \draw[-,dotted] (2,0) -- (2,1);
        \draw[-,dotted] (3,0) -- (3,1);
        \end{scope}
        \end{scope}
        
        \end{tikzpicture}
    \end{adjustbox}
    \caption{Dimer structure with a defect.}
    \label{fig: geometrical defect}
\end{figure}

For such systems the geometric symmetries mean that the capacitance matrix from \eqref{eq: general capacitance matrix} has the following tridiagonal block structure:

\setcounter{MaxMatrixCols}{20}
\begin{align}
\label{eq: strucutre capacitance matrix}
    \mathcal{C} = \begin{pNiceMatrix}
    \Block[draw,fill=blue!40,rounded-corners]{7-7}{} \tilde{\alpha} & \beta_{1} &&&&&&&&&\\
\beta_{1} & \alpha & \beta_{2}  &&&&&&&&\\
& \beta_{2} & \alpha & \beta_{1}  &&&&&&&\\
       && \ddots     & \ddots     & \ddots     &&&&&&\\
       &&& \beta_{2} & \alpha & \beta_{1}  &&&&&\\
       &&&& \beta_{1} & \alpha & \beta_{2}  &&&&\\
       &&&&& \beta_{2} & \Block[draw,fill=red!40,rounded-corners]{7-7}{}\eta & \beta_{2}  &&&\\
       &&&&&& \beta_{2} & \alpha & \beta_{1}  &&\\
       &&&&& && \beta_{1} & \alpha & \beta_{2}  &\\
       &&&&&&&& \ddots     & \ddots & \ddots     &\\
       &&&&&&&&& \beta_{1} & \alpha & \beta_{2}  \\
       &&&&&&&&& & \beta_{2} & \alpha & \beta_{1}  \\
       &&&&&&&&&&& \beta_{1} & \tilde{\alpha}
    \end{pNiceMatrix}
\end{align}
%
%
where
\begin{align}
\label{eq: translation alpha beta to s1 s2}
    \beta_1 = -s_1\inv,\quad \beta_2 = -s_2\inv,\quad \alpha = s_1\inv + s_2\inv,\quad \eta=2s_2\inv,\quad \tilde{\alpha} = s_1^{-1}.
\end{align}
It will be useful to have notation for the top left $(2m+1)\times(2m+1)$-submatrix $C_1
$ and the bottom right $(2m+1)\times(2m+1)$-submatrix $C_2
$ (as highlighted by the shading in \eqref{eq: strucutre capacitance matrix}). 
\cref{fig: eva eve} shows the spectrum and the eigenvectors of the capacitance matrix \eqref{eq: strucutre capacitance matrix}. It is immediately clear that there exists a spectral gap and there exists a localised eigenmode associated to an eigenvalue in the gap.

\begin{figure}[h]
    \centering
    \begin{subfigure}[t]{0.48\textwidth}
    \centering
    \includegraphics[height=0.76\textwidth]{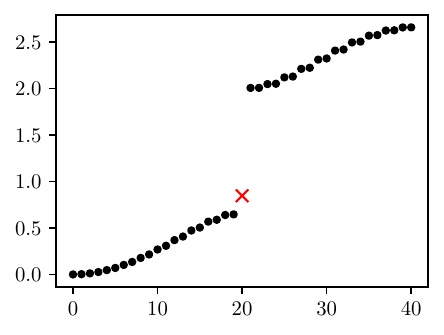}
    \caption{Eigenvalues of \eqref{eq: strucutre capacitance matrix}. As red cross a specific eigenvalue lying isolated from the others.}
    \label{fig: eva}
    \end{subfigure}\hfill
    \begin{subfigure}[t]{0.48\textwidth}
    \centering
    \includegraphics[height=0.76\textwidth]{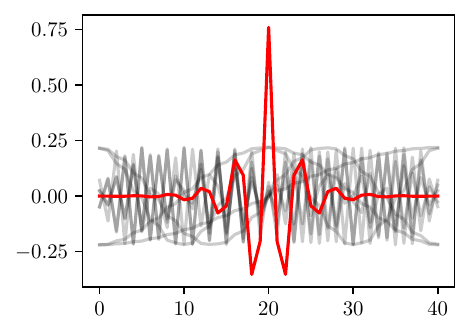}
    \caption{A selection of eigenvectors of $\mathcal{C}$ from \eqref{eq: strucutre capacitance matrix}. All eigenvectors are superimposed and with unit norm. The red solid eigenvector corresponds to the red cross eigenvalue in \eqref{fig: eva}.}
    \label{fig: eva eve}
    \end{subfigure}
    
    \caption{Eigenvalues and eigenvectors of the capacitance matrix \eqref{eq: strucutre capacitance matrix} for $N=41, s_1=1$ and $s_2=3$.}
    \label{fig: eigenvector behaviour}
\end{figure}

\section{Perturbed tridiagonal $2$-Toeplitz matrices}  \label{sect3}
In this section we will briefly recall results established in \cite{ammari2023perturbed} about eigenvalues and eigenvectors of tridiagonal $2$-Toeplitz matrices with perturbations on the corners.

Denote by 
\begin{equation}\label{equ:odd2toeplitzcornerperturbed1}
    A_{2k+1}^{(a, b)}(\alpha, \beta_1, \beta_2)\coloneqq\left(\begin{array}{ccccccc}
        \alpha+a & \beta_1 & 0 & 0 & \ldots & 0 & 0\\
        \beta_1 & \alpha & \beta_2 & 0 & \ldots & 0 & 0 \\
        0 & \beta_2 & \alpha & \beta_1 & \ldots & 0 & 0 \\
        \ldots & \ldots & \ldots & \ldots & \ldots & \ldots & \ldots \\
        \ldots & \ldots & \ldots & \ldots & \ldots & \alpha & \beta_2 \\
        0 & 0 & 0 & 0 & \ldots & \beta_2 & \alpha+b
    \end{array}\right) \in \R^{(2k+1)\times (2k+1)}
\end{equation}
and
\begin{equation}\label{equ:even2toeplitzcornerperturbed1}
    A_{2k}^{(a, b)}(\alpha, \beta_1, \beta_2)\coloneqq\left(\begin{array}{ccccccc}
        \alpha+a & \beta_1 & 0 & 0 & \ldots & 0 & 0\\
        \beta_1 & \alpha & \beta_2 & 0 & \ldots & 0 & 0 \\
        0 & \beta_2 & \alpha & \beta_1 & \ldots & 0 & 0 \\
        \ldots & \ldots & \ldots & \ldots & \ldots & \ldots & \ldots \\
        \ldots & \ldots & \ldots & \ldots & \ldots & \alpha & \beta_1 \\
        0 & 0 & 0 & 0 & \ldots & \beta_1 & \alpha+b
    \end{array}\right)\in \R^{2k\times 2k}
\end{equation}
the prototypical tridiagonal $2$-Toeplitz matrices with perturbations on the corners.

\subsection{Eigenvalues}
We define the polynomials 
\begin{equation}\label{equ:defiofpkstar1}
P_k^*(x):=\left(\beta_1 \beta_{2}\right)^k U_k\left(\frac{(x-\alpha)^2-\beta_1^2-\beta_2^2}{2 \beta_{1} \beta_2}\right), 
\end{equation}
where $U_k$ is the Chebyshev polynomial of the second kind, as well as the function
\begin{align}
\label{eq: def y}
    y(z) \coloneqq \frac{z^2-\beta_1^2-\beta_2^2}{2 \beta_{1} \beta_2}.
\end{align}
The following proposition holds. 
\begin{proposition}[Eigenvalues]
    The characteristic polynomials of $A_{2k+1}^{(a, b)}$ and $A_{2k}^{(a, b)}$ are respectively 
\begin{align}\label{equ:eigenpolynomial3}
\chi_{A_{2k+1}^{(a, b)}}(x) = \left(x-\alpha-a-b\right) P_k^*\left(x\right)+\left( ab\left(x-\alpha\right)  -a \beta_{1}^2-b \beta_{2}^2\right) P_{k-1}^*\left(x\right)
\end{align}
and 
\begin{align}\label{equ:eigenpolynomial4}
\chi_{A_{2k}^{(a, b)}}(x)=P_k^*\left(x\right)+\left((a+b)\left(\alpha-x\right)+ab+\beta_{2}^2\right) P_{k-1}^*\left(x\right)+a b \beta_{1}^2 P_{k-2}^*\left(x\right).
\end{align}
\end{proposition}

\subsection{Eigenvectors}


We start by defining two families of polynomials $\widehat p_{k+1}^{(\xi_{p}, \xi_{q})}(x)$ and $\widehat q_{k+1}^{(\xi_{p}, \xi_{q})}(x)$ as solutions to the recursion relations

\begin{equation}
\begin{aligned}\label{equ:Chebyshevrecurrence1}
 &\widehat p_{0}^{(\xi_{p}, \xi_{q})}(\mu) =\xi_{p},\quad  \widehat p_{1}^{(\xi_{p}, \xi_{q})}(\mu) =2\mu\xi_{p}+ \frac{\xi_{p}-\xi_{q}}{\beta},\\
 &\widehat p_{k+1}^{(\xi_{p}, \xi_{q})}(\mu)=2\mu \widehat p_k^{(\xi_{p}, \xi_{q})}(\mu)-\widehat p_{k-1}^{(\xi_{p}, \xi_{q})}(\mu),
\end{aligned}
\end{equation}
and
\begin{equation}
\begin{aligned}\label{equ:Chebyshevrecurrence2}
 & \widehat q_{0}^{(\xi_{p}, \xi_{q})}(\mu) =\xi_{q},\quad  \widehat q_{1}^{(\xi_{p}, \xi_{q})}(\mu) =(2\mu+\beta)\xi_{p}+\frac{\xi_{p}-\xi_{q}}{\beta},\\
 &\widehat q_{k+1}^{(\xi_{p}, \xi_{q})}(\mu)=2\mu \widehat q_k^{(\xi_{p}, \xi_{q})}(\mu)-\widehat q_{k-1}^{(\xi_{p}, \xi_{q})}(\mu),
\end{aligned}
\end{equation}
where $\beta=\beta_2/\beta_1$.
Then we state the following result. 
\begin{proposition}\label{thm: eigenvectors of A2k+1 and A2k}
Let $\lambda$ be an eigenvalue of $A_{2k+1}^{(a, b)}(\alpha, \beta_1, \beta_2)$. Then, using $\mu\coloneqq y(\lambda)$ where $y$ is as in \eqref{eq: def y}, the eigenvector corresponding to $\lambda$ is given by
\begin{align}\label{eq: eigenvector A2k+1}
\bm x =& \left( \hat q_{0}^{(\xi_p, \xi_{q})}\left(\mu\right),-\frac{1}{\beta_1}\left(\alpha-\lambda\right) \hat p_0^{(\xi_p, \xi_{q})}\left(\mu\right), \hat  q_1^{(\xi_p, \xi_{q})}\left(\mu\right), -\frac{1}{\beta_1} \left(\alpha-\lambda\right) \hat p_1^{(\xi_p, \xi_{q})}\left(\mu\right), \right. \nonumber\\
	&\qquad  \left.  \ldots, -\frac{1}{\beta_1} \left(\alpha-\lambda\right) \hat p_{k-1}^{(\xi_p, \xi_{q})}\left(\mu\right), 
 \hat  q_{k}^{(\xi_p, \xi_{q})}\left(\mu\right)\right ).
\end{align}
If $\lambda$ is an eigenvalue of $A_{2k}^{(a, b)}(\alpha, \beta_1, \beta_2)$, then the corresponding eigenvector is given by 
\begin{align}\label{eq: eigenvector A2k}
    \bm x =& \left( \hat q_{0}^{(\xi_p, \xi_{q})}\left(\mu\right),-\frac{1}{\beta_1}\left(\alpha-\lambda\right) \hat p_0^{(\xi_p, \xi_{q})}\left(\mu\right), \hat  q_1^{(\xi_p, \xi_{q})}\left(\mu\right), -\frac{1}{\beta_1} \left(\alpha-\lambda\right) \hat p_1^{(\xi_p, \xi_{q})}\left(\mu\right), \right. \nonumber\\
	&\qquad  \left.  \ldots, -\frac{1}{\beta_1} \left(\alpha-\lambda\right) \hat p_{k-1}^{(\xi_p, \xi_{q})}\left(\mu\right)\right ).
\end{align}
In both cases, $\xi_{q}=(\alpha-\lambda), \xi_{p}= (\alpha+a-\lambda)$.
\end{proposition}

The last result of this subsection is on the structure of the eigenvectors for the capacitance matrix $\mathcal{C}$ defined by \eqref{eq: strucutre capacitance matrix}. \cref{thm: eigenvectors of defect C} is an adapted version of \cite[Theorem 4.3]{ammari2023perturbed}.

\begin{proposition}\label{thm: eigenvectors of defect C}
Let $(\lambda,\bm v)$ be an eigenpair of $\mathcal{C}$ and let $\mu\coloneqq y(\lambda)$. Then $\bm v$ is given by
\begin{equation}\label{eq: structure eigenvector defect matrix}
\bm v = (\bm x^{(1)},\bm x^{(2)},\dots,\bm x^{(2m)},\bm x^{(2m+1)},(-1)^\sigma\bm x^{(2m)},\dots,(-1)^\sigma\bm x^{(2)},(-1)^\sigma\bm x^{(1)})^{\top},
\end{equation}
where $\bm x\in \R^{2m+1}$ is as in \eqref{eq: eigenvector A2k+1} with $\xi_{q}=(\alpha-\lambda), \xi_{p}= (\alpha+a-\lambda)$ and $\sigma\in\{0,1\}$ except for $\bm x \in \sspan \{\bm 1\}$ where $\sigma=1$.  
\end{proposition}
\begin{proof}
We will show \eqref{eq: structure eigenvector defect matrix} by showing that $(\mathcal{C}-\lambda I)\bm v=\bm 0$ in two steps.\\
\textbf{Step 1.} In the first step, we consider the first $2m+1$ rows of $(\mathcal{C}-\lambda I)\bm v$:
\begin{equation}\label{eq: eigenvectors C setp 1 base}
\left(\begin{array}{cccccc}
  \alpha+a-\lambda  & \beta_{1} & & &&\\
  \beta_{1} & \alpha-\lambda  & \ddots & &&\\
  & \ddots & \ddots&\ddots&&\\
  &  &\beta_{1} & \alpha-\lambda  & \beta_2&\\
  &  & & \beta_{2} & \eta-\lambda&\beta_2
\end{array}\right)\begin{pmatrix}
\bm x\\
\bm * 
\end{pmatrix}= \bm 0,
\end{equation}
where we put $\bm *$ instead of $\bm x^{(2m)}$ to show that we will not use this entry in what follows as we only need the first $2m$ rows of \eqref{eq: eigenvectors C setp 1 base} to determine the $2m+1$ entries of $\bm x$. We write $\bm x$ as
\begin{align}\label{equ:proofeigenvectoreq1}
    \bm x=\left(x_1,\ -\frac{1}{\beta_1}(\alpha-\lambda)x_2,\ \frac{\beta_1}{\beta_2}x_3,\ \cdots,\ -\frac{1}{\beta_1}\left(\frac{\beta_1}{\beta_2}\right)^{m-1} \left(\alpha-\lambda\right)x_{2m}, \left(\frac{\beta_1}{\beta_2}\right)^{m}x_{2m+1}\right).
\end{align}
Considering the first row of \eqref{eq: eigenvectors C setp 1 base}, we can choose
\begin{align*}
    x_1 =(\alpha-\lambda), \quad x_2 =  (\alpha+a -\lambda).
\end{align*}
Then by the second row, we have 
\begin{align*}
    \beta_{1} x_1 -\frac{1}{\beta_{1}} (\alpha-\lambda)^2 x_2 +\beta_1x_3 =0,   
\end{align*}
 which gives 
\[
x_3 =  \frac{(\alpha-\lambda)^2}{\beta_1^2}x_2-x_1.
\]
The third row is
\[
-\frac{\beta_{2}}{\beta_{1}} (\alpha-\lambda)x_2 +(\alpha-\lambda)\left(\frac{\beta_1}{\beta_2}\right)x_3 -(\alpha-\lambda)\left(\frac{\beta_1}{\beta_2}\right)x_4=0,
\] 
and thus, 
\begin{align*}
x_4 = - \frac{\beta_{2}^2}{\beta_{1}^2}x_2+x_3.
\end{align*}
Continuing the process, we can easily verify that 
\begin{align}\label{eq: recurrence relations xis}
x_{2k+1} &= \frac{(\alpha-\lambda)^2}{\beta_1^2} x_{2k}- x_{2k-1}, \quad 1\leq k\leq m,\\
x_{2k}&=x_{2k-1} -\frac{\beta_{2}^2}{\beta_{1}^2} x_{2k-2},\quad \quad 2\leq k\leq m.
\end{align}
Applying to \eqref{eq: recurrence relations xis} the same manipulations as those used in the  proof of \cite[Theorem 3.3]{ammari2023perturbed}, it is now easy to see that $\bm x$ is of the form \eqref{eq: eigenvector A2k+1}.

\textbf{Step 2.} We conclude the proof by considering the last $2m$ rows. By the fact that $\mathcal{C}$ is $S=\antidiag(1,\dots,1)\in\R^{(4m+1)\times (4m+1)}$ symmetric, \emph{i.e.}, $S\mathcal{C}S=\mathcal{C}$, and $\mathcal{C}$ has only simple eigenvalues \cite{feppon.cheng.ea2023Subwavelength} we can immediately conclude that $S\bm v = \pm \bm v$ and thus obtain the desired result. 
\end{proof}


\section{Asymptotic spectral gap and localised interface modes}  \label{sect4}
In this section we will prove the existence of a spectral gap for a defectless structure of dimers. Furthermore, we will demonstrate that a direct relationship exists between an eigenvalue being within the spectral gap and the localisation of its corresponding eigenvector.
\subsection{Asymptotic spectral gap}
We first show the existence of a spectral gap for the capacitance matrix of an unperturbed structure of dimers.

By physical considerations we assume $\beta_1,\beta_2<0$ and $\beta_1<\beta_2$. It will often be useful to shift $\mathcal{C}$ in order to have most diagonal entries zero, to that end we introduce $z(x)\coloneqq x - \alpha$.

\begin{definition}[Spectral bulk and gaps] \label{def:spectralgap}
    Consider a finite structure of resonators. We define the \emph{asymptotic spectral bulk} $\Sigma$ and \emph{asymptotic spectral gap} $\Gamma$ of the structure as the spectral bulk and spectral gap (also known as band gap) of the associated infinite periodic system, respectively.
\end{definition}
The spectral gap and spectral bulk of infinite periodic dimer systems have been computed in \cite[Lemma 5.3]{ammari.ea2023Edge}.
\begin{proposition}
    Consider a system of repeated dimers (without defect) with $N=2m$ resonators. Denote by $C_{N}$ the associated capacitance matrix and let $\Sigma$ be the asymptotic spectral bulk. Then 
    \begin{align*}
        \Sigma = \overline{\lim_{N\to\infty} \sigma(C_{N})} = \left[0,\frac{2}{s_2}\right] \cup \left[\frac{2}{s_1}, \frac{2}{s_1} + \frac{2}{s_2}\right],
    \end{align*}
    where $\lim$ denotes the Hausdorff limit. Consequently, the asymptotic spectral gap is
    \begin{align*}
        \Gamma = \left( \frac{2}{s_2}, \frac{2}{s_1} \right)\subset \R.
    \end{align*}
\end{proposition}
\begin{proof}
From \eqref{equ:eigenpolynomial4} we see that any eigenvalue $\lambda$ of the capacitance matrix of a structure of dimer without defect --- that is, a capacitance matrix of the form $C_N = A_{2m}^{(\beta_2,\beta_2)}$ --- satisfy   
\begin{align*}
    0 &= P^*_m(\lambda) + (-2\beta_2z+2\beta_2^2)P_{m-1}^*(\lambda)+\beta_2^2\beta_1^2P^*_{m-2}(\lambda)\\
    &= \beta_1^m\beta_2^m U_m(y) + (-2\beta_2z+2\beta_2^2)\beta_1^{m-1}\beta_2^{m-1} U_{m-1}(y)+\beta_1^{m}\beta_2^{m} U_{m-2}(y),
\end{align*}
where for ease of notation we use $z$ for $z(\lambda)$ and $y$ for $y(z(\lambda))$. Using the recurrence relation of the Chebyshev polynomials of the second kind we get that
\begin{align*}
0 &= \beta_1^m\beta_2^m[2yU_{m-1}(y) - U_{m-2}(y)] \\&+ (-2\beta_2z+2\beta_2^2)\beta_1^{m-1}\beta_2^{m-1} U_{m-1}(y)+\beta_1^{m}\beta_2^{m} U_{m-2}(y)\\
\Leftrightarrow 0 &= U_{m-1}(y) [(-2\beta_2z+2\beta_2^2)\beta_1^{m-1}\beta_2^{m-1} + \beta_1^m\beta_2^m2y]\\
\Leftrightarrow 0 &= U_{m-1}(y) [(-z+\beta_2) + \beta_1y].
\end{align*}
So the eigenvalues are given by
\begin{align}
    \lambda= \alpha+(\beta_1+\beta_2),\quad \lambda = \alpha+(\beta_2-\beta_1)\quad \text{and} \quad \lambda = \alpha \pm \sqrt{\beta_1^2+2\beta_1\beta_2\cos\left(\frac{k\pi}{m}\right)+\beta_2^2},
\end{align}
for $1\leq k\leq m-1$. The sought result follows now 
directly from inserting the definition of $\alpha, \beta_1,\beta_2$ in terms of $s_1$ and $s_2$ as in \eqref{eq: translation alpha beta to s1 s2}.
\end{proof}
We will refer to spectral bulk and spectral gap of the physical system and of the capacitance matrix interchangeably, using the identification \cref{prop: reduction to capacitance matrix}.
\subsection{Localised interface modes}
In this subsection, we will prove that an eigenvalue in the gap is associated with an eigenvector that is exponentially localised in the sense that it decays exponentially away from the interface in both directions, within the limits of the finite structure.

\begin{definition}[Localised interface mode]
    Let $v(x)$ be an eigenmode. Then we say that $v$ is a \emph{localised interface mode} at the point $x_0$, if both $\vert v(x-x_0)\vert $ for $x_0<x\in D$ and $\vert v(x_0-x)\vert $  for $x_0>x\in D$ decay exponentially as a function of $x\in D$. The same terminology applies to the corresponding eigenvector of the capacitance matrix.
\end{definition}

\begin{proposition}[Eigenvectors of $\mathcal{C}$]\label{prop: exponential decay and sines}
    Let $\mathcal{C}\in\mathbb{R}^{4m+1 \times 4m+1}$ be the capacitance matrix of a defect structure as illustrated in \cref{fig: geometrical defect} and let $(\lambda,v)$ be an eigenpair of $\mathcal{C}$. Then, there exists $\vert r\vert \geq1$ independent of $m$ and $A,B,\tilde A,\tilde B\in\R$ dependent on $m$ such that
    \begin{description}
        \item[if $y(\lambda)^2>1$]
    \begin{align*}
        v^{ (\vert 2m-2j\vert) } = Ar^j+Br^{-j},\\
        v^{ (\vert 2m-2j-1\vert) } = \tilde A r^j+\tilde B r^{-j};
    \end{align*}
    with $A=\frac{r^{1-m} (c_1 r-c_2)}{r^2-1}=\mathcal{O}(\frac{1}{r^{m}})$ and $B = \frac{r^m (c_2 r-c_1)}{r^2-1}= \mathcal{O}(r^{m-1})$ as $m\to\infty$ for $c_1,c_2\in\R$ independent of $m$. The same asymptotics (with a slight different formula) hold for $\tilde A$ and $\tilde B$;
        \item[if $y(\lambda)^2<1$]
        \begin{align*}
        v^{ (\vert 2m-2j\vert) } = A\cos(j\theta)+B\sin(j\theta),\\
        v^{ (\vert 2m-2j-1\vert) } = \tilde A\cos(j\theta)+\tilde B\sin(j\theta),
    \end{align*}
    with $r=e^{\mathbf{i} \theta}$ and $A,B,\tilde A,\tilde B$ bounded as $m\to\infty$;
    \item[if $y(\lambda)^2=1$] $r=\pm 1$ and
        \begin{align*}
        v^{ (\vert 2m-2j\vert) } = Ar_1^j+Br_1^j\cdot j,\\
        v^{ (\vert 2m-2j-1\vert) } = \tilde A r_1^j+\tilde B r_1^j\cdot j,
    \end{align*}
    with
     $A=\frac{r^{1-m} (c_1 m r-c_1 r-c_2 m)}{m r^2-m-r^2}$ and $B = \frac{r^m (c_2 r-c_1)}{m r^2-m-r^2}$ as $m\to\infty$ for $c_1,c_2\in\R$ independent of $m$. The same asymptotics (with a slight different formula) hold for $\tilde A$ and $\tilde B$.
    \end{description}
\end{proposition}
\begin{proof}
    It is enough to consider the $\widehat p_k^{(\xi_{p}, \xi_{q})}$ and $\widehat q_k^{(\xi_{p}, \xi_{q})}$ part of the eigenvector in \cref{thm: eigenvectors of A2k+1 and A2k} as the rest is a multiplicative factor. We only consider $\widehat p_k^{(\xi_{p}, \xi_{q})}$ as the procedure for $\widehat q_k^{(\xi_{p}, \xi_{q})}$ is exactly the same.
    Let $(\lambda,v)$ be an eigenpair and recall from \eqref{equ:Chebyshevrecurrence1} the recurrence formula
    \begin{align*}
        \widehat p_{k+1}^{(\xi_{p}, \xi_{q})}(\mu) = 2\mu\widehat p_k^{(\xi_{p}, \xi_{q})}(\mu) - 
        \widehat p_{k-1}^{(\xi_{p}, \xi_{q})}(\mu),
    \end{align*}
     where $\mu=y(\lambda)$. This is a two term recurrence formula with characteristic equation 
    \begin{align}\label{eq:characteristic equation}
        X^2-2\mu X + 1 =0
    \end{align}
    having none, one or two roots $r_1$ and $r_2$ in exactly the cases delineated in the proposition. The formulas for the eigenvectors follow. By Vieta's formula, we know that the two roots of \eqref{eq:characteristic equation} satisfy $r_1r_2=1$ so one of them must satisfy $\vert r_i \vert \geq 1$. The constants $A,B,\tilde A,\tilde B\in\R$ follow from the initial conditions \cref{equ:Chebyshevrecurrence1} with the values expressed in \cref{thm: eigenvectors of defect C}.
\end{proof}

The eigenvector in the case when $y(\lambda)^2>1$ is exponentially localized in the interface, as we can rescale the eigenvector to make
\begin{align*}
v^{ (\vert 2m-2j\vert) } = Br^{-j}+ Ar^j,\\
v^{ (\vert 2m-2j-1\vert) } = \tilde B r^{-j}+ \tilde A r^j, 
\end{align*}
where $\mathcal{O}(B) = \mathcal{O}(\tilde B)=\mathcal{O}(1)$ and $\mathcal {O}(Ar^j)=\mathcal {O}(\tilde Ar^j)=o(\frac{1}{r^{m-1}}), j=1,\cdots, 2m$. 

We remark that the three cases presented in \cref{prop: exponential decay and sines} correspond to $\lambda \in \Gamma$, $\lambda \in \Sigma$ and $\lambda \in \partial \Sigma$, respectively. We show this behaviour in \cref{fig: eigenvector behaviour}.

\begin{figure}[h]
    \centering
    \begin{subfigure}[t]{0.32\textwidth}
    \centering
    \includegraphics[height=0.76\textwidth]{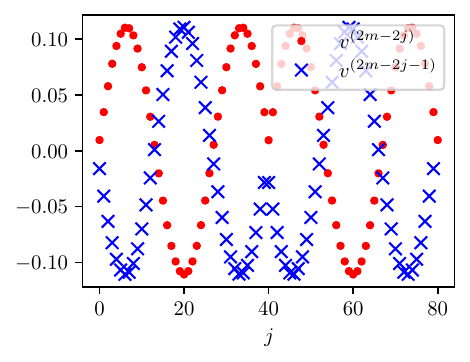}
    \caption{Eigenvector associated to an eigenvalue in the asymptotic spectral bulk.}
    \label{fig: eve eva in bulk}
    \end{subfigure}\hfill
    \begin{subfigure}[t]{0.32\textwidth}
    \centering
    \includegraphics[height=0.76\textwidth]{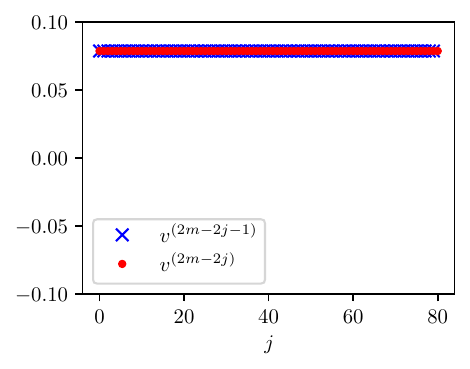}
    \caption{Eigenvector associated to an eigenvalue in the boundary of the asymptotic spectral bulk.}
    \label{fig: eve eva in bulk bdy}
    \end{subfigure}\hfill
    \begin{subfigure}[t]{0.32\textwidth}
    \centering
    \includegraphics[height=0.76\textwidth]{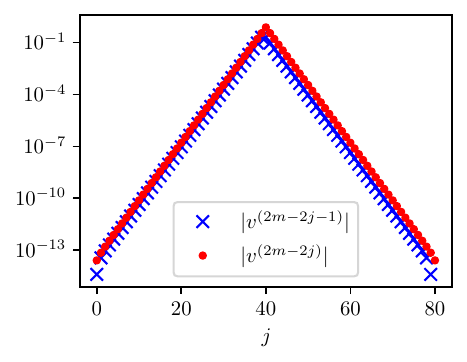}
    \caption{Eigenvector associated to an eigenvalue in the asymptotic spectral gap. $y$ axis in \textit{log}-scale.}
    \label{fig: eve eva in gap}
    \end{subfigure}
    
    \caption{Eigenvector behaviour based on the location of eigenvalue. Computation performed with $N=81, s_1=1, s_2=3$.}
    \label{fig: eigenvector behaviour2}
\end{figure}
In order to show that a localised eigenvector exists, it is thus sufficient and necessary to prove the existence of an eigenvalue in the asymptotic spectral gap. We will do this in the next section.

\section{Existence, uniqueness, and convergence of the eigenvalue in the gap}  \label{sect5}
In this section we will prove the existence of an unique eigenvalue in the gap for the defect structure and consequently the existence of a unique localised interface eigenvector. Furthermore, we analyse the behaviour as the size of the system grows, specifically the limiting behaviour as $N\to\infty$. We will show that the eigenvalue of $\mathcal{C}$ lying in the asymptotic spectral gap converges exponentially fast to a value in the gap.

\subsection{Existence}
We first show the existence of the eigenvalue in the band gap. We start by remarking that by performing Laplace expansion on the top block of rows $1\leq i\leq 2m+1$, we obtain that the determinant of $\mathcal{C}\in\R^{(4m+1)\times (4m+1)}$ is given by
    \begin{align}
        \det \mathcal{C} &=  \det A_{2m+1}^{(a,\eta-\alpha)}(\alpha,\beta_1,\beta_2)\cdot\det A_{2m}^{(0,b)}(\alpha,\beta_1,\beta_2)\nonumber\\
        &- \beta_2^2\det A_{2m}^{(a,0)}(\alpha,\beta_1,\beta_2)\cdot\det A_{2m-1}^{(0,b)}(\alpha,\beta_2,\beta_1),
    \end{align}
where $a=b=\beta_2$. We observe that the last term has $\beta_1$ and $\beta_2$ switched.
Consequently, the characteristic polynomial $p(x)$ of $\mathcal{C}$ is given by
\begin{align}
p(x) &= (\left[\left(x-\alpha-\beta_2-(\beta_1 - \beta_2)\right) P_m^*\left(x\right)+\left( \beta_2(\beta_1 - \beta_2)(x-\alpha)  -\beta_2 \beta_{1}^2-(\beta_1 - \beta_2) \beta_{2}^2\right) P_{m-1}^*\left(x\right)\right] \nonumber\\
&- \beta_2^2\left[\left(x-\alpha-\beta_2\right) P_{m-1}^*\left(x\right)+\left( -\beta_2 \beta_{1}^2\right) P_{m-2}^*\left(x\right)\right])\chi_{A_{2m}^{(a,0)}(\alpha,\beta_1,\beta_2)}(x),
\end{align}
as one swiftly remarks that $A_{2m}^{(0,b)}(\alpha,\beta_1,\beta_2)$ and $A_{2m}^{(a,0)}(\alpha,\beta_1,\beta_2)$ are similar matrices.

For the sake of brevity, we rewrite
\begin{align*}
    p(x) &= \chi_{A_{2m}^{(a,0)}(\alpha,\beta_1,\beta_2)}(x)\left(\left[A P_m^*\left(x\right)+B P_{m-1}^*\left(x\right)\right] 
-\beta_2^2
\left[E P_{m-1}^*\left(x\right)+F P_{m-2}^*\left(x\right)\right]\right),
\end{align*}
where $A,B,E,F$ do not depend on $m$. Thus, $\lambda\in\Gamma$ is an eigenvalue if and only if
\begin{align}
    \left[A P_m^*\left(\lambda\right)+B P_{m-1}^*\left(\lambda\right)\right]
    &= \beta_2^2
    \left[E P_{m-1}^*\left(\lambda\right)+F P_{m-2}^*\left(\lambda\right)\right]\\
    \Leftrightarrow P_m^*\left(\lambda\right) \left[A+B\frac{P_{m-1}^*\left(\lambda\right)}{P_{m}^*\left(\lambda\right)}\right]
    &= \beta_2^2
    P_{m-1}^*\left(\lambda\right) \left[E+F\frac{P_{m-2}^*\left(\lambda\right)}{P_{m-1}^*\left(\lambda\right)}\right].
    \label{eq: lambda in gamma eigenvalues only if}
\end{align}
In the above step we were able to divide by $\chi_{A_{2m}^{(a,0)}}(\lambda)$ and $P_{m-2}^*\left(\lambda\right)$ as we only consider $\lambda\in\Gamma$ so the latter term is non-zero because the Chebyshev polynomials only have roots in $(-1,1)$ and the former term was shown to be non-zero in $\Gamma$ in \cref{prop: exaclty one in gap}. Moreover, since 
\[
U_m(y)=\frac{\left(y+\sqrt{y^2-1}\right)^{m+1}-\left(y-\sqrt{y^2-1}\right)^{m+1}}{2 \sqrt{y^2-1}},
\] 
the limit  $L(\lambda)\coloneqq \lim_{m\to\infty}\frac{P_{m-1}^*\left(\lambda\right)}{P_{m}^*\left(\lambda\right)}$ exists for all $\lambda\in \mathbb R$. Then,
\begin{align*}
    \left[A+BL\right]
    = \beta_2^2
    L
    \left[E+FL\right]\\
    \Leftrightarrow L^2(\beta_2^2F) + L(\beta_2^2E-B) - A = 0, 
\end{align*}
from which we get the condition
\begin{align}\label{eq: condition L is root}
    L(\lambda) &= \frac{B - E \beta_{2}^{2} \pm \sqrt{4 A F \beta_{2}^{2} + B^{2} - 2 B E \beta_{2}^{2} + E^{2} \beta_{2}^{4}}}{2 F \beta_{2}^{2}}. 
\end{align}
On the other hand, by the recurrence formula of $U_m$, for $\lambda \in \Gamma$ where $y(z(\lambda))<-1$,
\begin{align}\label{eq: condition L is limit}
    L(\lambda) = \lim_{m\to\infty}\frac{P_{m-1}^*\left(\lambda \right)}{P_{m}^*\left(\lambda\right)} = \frac{y-\sqrt{y^2-1}}{\beta_1\beta_2}=\frac{z^2-\beta_1^2-\beta_2^2-2 \beta_1 \beta_2 \sqrt{\frac{\left(\beta_1^2+\beta_2^2-z^2\right)^2}{4 \beta_1^2
   \beta_2^2}-1}}{2 \beta_1^2 \beta_2^2},
\end{align}
by using again the shorthand $z$ for $z(\lambda)$ and $y$ for $y(z(\lambda))$. Now, an algebraic manipulation shows that conditions \eqref{eq: condition L is root} and \eqref{eq: condition L is limit} have exactly one common solution $\lambda_0 \in \Gamma$  given by
\begin{align}\label{eq: root of L}
    z(\lambda_0) = \frac{1}{2} \left(-\sqrt{9 \beta_1^2-14 \beta_1 \beta_2+9 \beta_2^2}-\beta_1-\beta_2\right).
\end{align}

\begin{proposition}\label{prop:existencedefectfrequency}
    Consider a perturbed structure of dimers as illustrated in \cref{fig: geometrical defect}. For $N$ large enough there exists at least one localised interface eigenvector of $\mathcal{C}$ with eigenvalue $\lambda_{\mathsf{i}}^{(N)}$ in the band gap $\Gamma$.
\end{proposition}
\begin{proof}
Denote by 
\begin{align*}
    f_m(\lambda) \coloneqq &\left[A+B\frac{P_{m-1}^*\left(\lambda\right)}{P_{m}^*\left(\lambda\right)}\right] - \beta_2^2\frac{P_{m-1}^*\left(\lambda\right)}{P_m^*\left(\lambda\right)}
     \left[E+F\frac{P_{m-2}^*\left(\lambda\right)}{P_{m-1}^*\left(\lambda\right)}\right],\\
     f_{\infty} \coloneqq & \left[A+BL(\lambda)\right] - \beta_2^2L(\lambda)
     \left[E+FL(\lambda)\right], 
\end{align*}
and note that $f_\infty(\lambda)=\lim_{m\to\infty}f_m(\lambda)$. By (\ref{eq: root of L}) there exists a $\lambda_0$ in the band gap such that $f_\infty(\lambda_0)=0$. Furthermore, by the above formula of $f_\infty(\lambda)$, we have $f_{\infty}(\lambda_0-\zeta)f_{\infty}(\lambda_0+\zeta)<0$ for some $\zeta>0$ satisfying $[\lambda_0-\zeta, \lambda_0+\zeta]\subset \Gamma$. By the sign-preserving property, we have $f_{m}(\lambda_0-\zeta)f_{m}(\lambda_0+\zeta)<0$ for large enough $m$, which proves the existence of roots of $f_{m}(\lambda)$ in the band gap $\Gamma$.  By \cref{prop: exponential decay and sines} the corresponding eigenvector is a localised interface eigenvector.
\end{proof}



\subsection{Uniqueness}
We now show the uniqueness of the eigenvalue in the band gap. We will use the following result about the monotonicity of Chebyshev polynomials of the second kind.
\begin{lemma}\label{lemma: monotonicity_Chebyshev}
Let $k\in\N$, then
\begin{align*}
    \frac{U_{k-1}(x)}{U_k(x)}
\end{align*}
is strictly decreasing for $x\in(-\infty, -1)\cup(1,+\infty)$ for any $k\in\N$.
\end{lemma}
\begin{proof}
Considering the derivative of $\dfrac{U_{k-1}(x)}{U_k(x)}$, we have 
\begin{equation}\label{equ:chebmonotoneequ1}
\left(\frac{U_{k-1}(x)}{U_k(x)}\right)\pri= \frac{U_{k}(x)U_{k-1}\pri(x)-U_{k}\pri(x)U_{k-1}(x)}{U_{k}^2(x)}.
\end{equation}
A well-known property of Chebyshev polynomials is that
\[
U_k\pri(x)=\frac{(k+1) T_{k+1}(x)-x U_k(x)}{x^2-1},
\]
where $T_{k}(x)$ are the Chebyshev polynomials of the first kind. Thus, to demonstrate the negativity of (\ref{equ:chebmonotoneequ1}) for $\vert x\vert>1$, we only need to show that

\begin{align}\label{eq: condition negativity}
    0 &> U_k(x)[k T_{k}(x)-x U_{k-1}(x)] - [(k+1) T_{k+1}(x)-x U_k(x)]U_{k-1}(x)\nonumber \\&= U_k(x)k T_{k}(x) - (k+1) T_{k+1}(x)U_{k-1}(x).
\end{align}
It is also well-known that 
\begin{equation}\label{equ:chebyshevrelation1}
T_{\ell}(x) U_n(x)= \begin{cases}\frac{1}{2}\left(U_{\ell+n}(x)+U_{n-\ell}(x)\right), & \text { if } n \geq \ell-1, \\ \frac{1}{2}\left(U_{\ell+n}(x)-U_{\ell-n-2}(x)\right), & \text { if } n \leq \ell-2 . \end{cases}
\end{equation}
So that \eqref{eq: condition negativity} becomes
\begin{align*}
   \frac{k}{2}(U_{2k}(x)+U_{0}(x))-\frac{k+1}{2}(U_{2k}(x)-U_{0}(x)) = -\frac{U_{2k}(x)-(2k+1)}{2}.
\end{align*}
By $U_{2k}(x)>2k+1, k>0, x\in (-\infty, -1)\cup(1,+\infty)$ and the proof is complete.
\end{proof}
\begin{proposition}\label{prop: exaclty one in gap}
    There exists at most one eigenvalue of $\mathcal{C}$ as defined in \eqref{eq: strucutre capacitance matrix} lying in the asymptotic spectral gap $\Gamma = (2/s_2,2/s_1)$. In particular, for $m$ large enough, there exists exactly one eigenvalue in $\Gamma$.  
\end{proposition}
\begin{proof}
Considering the compression of $\mathcal{C}$ obtained by removing the central row and column we obtain a block diagonal matrix with blocks given by
\begin{align*}
    B = \begin{pmatrix}
    \alpha + \beta_2 & \beta_{1}\\
\beta_{1} & \alpha & \beta_{2}\\
& \beta_{2} & \alpha & \beta_{1}\\
       && \ddots     & \ddots     & \ddots\\
       &&& \beta_{2} & \alpha & \beta_{1}  \\
       &&&& \beta_{1} & \alpha \\
    \end{pmatrix}\quad \text{and}\quad PBP,
    \end{align*}
where $P = \antidiag(1,\dots,1)$. Remark that the same assumptions as stated before \cref{def:spectralgap} on the coefficients of $B$ hold, except that now $b=0$, that is, the matrix $B$ is of the type $A^{(a,0)}_{2m}$ and its characteristic polynomials is thus given by 
$$
P_m^*\left(x\right)+\left(-\beta_2z+\beta_{2}^2\right) 
P_{m-1}^*\left(x\right).
$$
Consequently, we see from \eqref{equ:defiofpkstar1} that $\lambda$ is an eigenvalue of $B$ if and only if
\begin{align}\label{eq: condition eigenvalue block}
    \underbrace{\frac{-\beta_1\beta_2}{-\beta_2z+\beta_2^2}}_{\eqqcolon L(z)}=\underbrace{\frac{U_{m-1}(y)}{U_m(y)}}_{\eqqcolon R(z)},
\end{align}
by using again the shorthand $z$ for $z(\lambda)$ and $y$ for $y(z(\lambda))$. Remark that $L(z)$ in \eqref{eq: condition eigenvalue block} is strictly monotonous increasing with a pole of order one at $\beta_2$. Furthermore, by \cref{lemma: monotonicity_Chebyshev} $R(z)$ is strictly monotonically increasing on $(\beta_1-\beta_2,0)$ and strictly monotonically decreasing on $(0,\beta_2-\beta_1)$

We first consider the case $2\beta_2 < \beta_1 < \beta_2$, which equates to the pole of $L(z)$ lying outside of the gap $\Gamma$. In this case we have
\begin{align*}
    L(\beta_1-\beta_2) = - \frac{\beta_1}{2\beta_2-\beta_1} &< -\frac{m}{m+1} = R(\beta_1-\beta_2),\\
    L(\beta_2-\beta_1) = -1 &< -\frac{m}{m+1}  = R(\beta_2-\beta_1),
\end{align*}
so that $L(z)<R(z)$ for all $z\in(\beta_1-\beta_2,\beta_2-\beta_1)$.

The case $\beta_1 \leq 2\beta_2$ is similar. The pole $z=\beta_2$ is now in the interval of interest, but now $L(z)>0>R(z)$ for $z\in(\beta_1-\beta_2,\beta_2)$ and $R(z)<L(z)$ for $z\in(\beta_2, \beta_2-\beta_1)$ by the same argument as the one in the previous case. Thus there can be at most one eigenvalue in the gap. For $m$ large enough, by \cref{prop:existencedefectfrequency} there exists exactly one eigenvalue in the gap.
\end{proof}

\begin{remark}
We remark that if one can show the existence of an eigenvector in the band gap $\Gamma$ for the capacitance matrix $\mathcal C$ with general size $m$, then by \cref{prop: exaclty one in gap}, it is unique. 
\end{remark}

\subsection{Convergence}
Last, we show the exponential convergence of the resonant frequency in the gap and conclude the section by the following theorem.

\begin{theorem}\label{thm:existenceofeigenfrquency}
    Consider a perturbed structure of dimers as illustrated in \cref{fig: geometrical defect}. For $N$ large enough there exists a unique interface mode with eigenfrequency $\omega_{\mathsf{i}}^{(N)}$ in the band gap. The associated eigenfrequency $\omega_{\mathsf{i}}^{(N)}$ converges to 
    \begin{align}
        \omega_{\mathsf{i}} = v_b\sqrt{\delta \frac{1}{2} \left(-\sqrt{\frac{9}{s_1^2}- \frac{14}{s_1s_2} + \frac{9}{s_2^2}}+\frac{3}{s_1}+\frac{3}{s_2}\right)}
    \end{align}
    exponentially as $N\to\infty$. In particular, for $N$ big enough,
    \begin{align}\label{eq: error estimate convergence frequency in gap}
        \vert \omega_{\mathsf{i}} - \omega_{\mathsf{i}}^{(N)}\vert 
        < Ae^{-BN},
    \end{align}
    for some $A,B>0$ independent of $N$.
\end{theorem}
\begin{proof}
The limit $\lambda_i$ can be computed from (\ref{eq: root of L}). The only part left to prove is the convergence rate. Denote by $0_d\in\R^d$ the zero element of that vector space and by $v_{\mathsf{i},N}$ the $l^2$-normalised localised eigenvector associated to $\lambda_{\mathsf{i}}^{(N)}$. Then, we remark that
    \begin{align}
    \label{eq: vgN0 is pseudoeigenvector}
\left\Vert(\mathcal{C}_{N+4}-\lambda_{\mathsf{i}}^{(N)})\begin{pmatrix}
            0_2\\v_{\mathsf{i},N}\\0_2
        \end{pmatrix}\right\Vert_2 = \left\Vert\begin{pmatrix}
            0\\
            \beta_2 v_{\mathsf{i},N}^{(1)}\\
        -v_{\mathsf{i},N}^{(1)}\beta_2
            \\
            0_{N-6}\\
            -v_{\mathsf{i},N}^{(N)}\beta_2\\
            \beta_2 v_{\mathsf{i},N}^{(N)}\\
            0
        \end{pmatrix}\right\Vert_2 = 4\beta_2^2(v_{\mathsf{i},N}^{(1)})^2\eqqcolon \epsilon_{N}.
    \end{align}
    Remark also that the same estimation holds if we replace $\mathcal{C}_{N+4}$ in \eqref{eq: vgN0 is pseudoeigenvector} with $\mathcal{C}_{N+4k}$ for any $k\in\N$. This shows that $\lambda_{\mathsf{i}}^{(N)}$ is an $\epsilon$-pseudo-eigenvalue of $\mathcal{C}_{N+4k}$ for any $k\in\N$ for every $\epsilon>\epsilon_{N}$. Since $\mathcal{C}_N$ is normal, by Theorem \ref{thm:pseudospectraofnormaloperator}, we have $\vert \lambda_{\mathsf{i}} - \lambda_{\mathsf{i}}^{(N)}\vert\leq \epsilon_{N}$. On the other hand,  for $N$ big enough, $\lambda_{\mathsf{i}}^{(N)}$ is close to $\lambda_{\mathsf{i}}$ and the formulae of \cref{prop: exponential decay and sines} yield \begin{align}\label{eq: epsionN0 goes to 0}
        v_{\mathsf{i},N}^{(1)} = \mathcal{O}(1)\qquad \Vert v_{\mathsf{i},N} \Vert_2=\mathcal{O}(e^{C_1N})\qquad\text{as } {N\to\infty}
    \end{align}
    for some $C_1>0$. So if $v_{\mathsf{i},N}$ needs to be normalised $v_{\mathsf{i},N}^{(1)}$ decays exponentially. By (\ref{eq: vgN0 is pseudoeigenvector}), $\epsilon_N$ decays exponentially, which together with \cref{prop: reduction to capacitance matrix} proves \eqref{eq: error estimate convergence frequency in gap}.
 \end{proof}

We remark that, combined with Proposition \ref{prop: exponential decay and sines}, Theorem \ref{thm:existenceofeigenfrquency} also gives the decaying rate of the interface mode for a structure with sufficiently many resonators. 

\begin{figure}[h]
    \centering
    \includegraphics[width=0.45\textwidth]{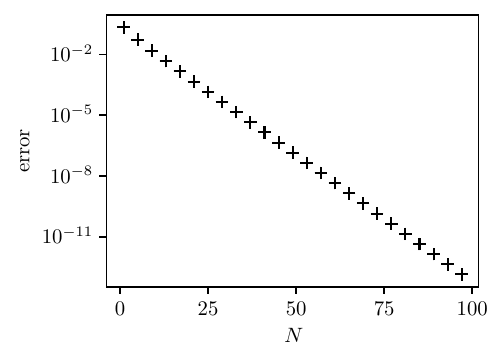}
    \caption{Convergence of the eigenvalue in the gap ($y$-axis in \emph{log} scale). We display the left-hand side of \eqref{eq: error estimate convergence frequency in gap} for a structure with $s_1=1$ and $s_2=2$.
    }
    \label{fig: convergence of error}
\end{figure}

\section{Stability analysis} \label{sect6}
Interface modes of SSH-like structures are well known to be stable, \emph{i.e.}, perturbations of the system affect them only slightly. In this section we show that perturbation in the geometry have limited effect on both the resonant frequencies and the associated modes and we quantify this effect. 

To this end we consider a system of $N=4m+1$ resonators as shown in \cref{fig: geometrical defect} but the spacings $s_i$ are now perturbed:
\begin{align}
    \label{eq: perturbation si-s}
    s_i = \begin{dcases}
        s_1 + \tilde\epsilon_i,& 1\leq i \leq 2m, \text{ $i$ odd},\\
        s_2 + \tilde\epsilon_i,& 1\leq i \leq 2m, \text{ $i$ even},\\
        s_1 + \tilde\epsilon_i,& 2m+1\leq i \leq 4m, \text{ $i$ even},\\
        s_2 + \tilde\epsilon_i,& 2m+1\leq i \leq 4m, \text{ $i$ odd}.
    \end{dcases}
\end{align}
Furthermore, we denote 
\begin{align}
    \label{eq: eps_i}
    \epsilon_i = \begin{dcases}
        -\frac{\tilde\epsilon_i}{s_1(s_1+\tilde\epsilon_i)},& 1\leq i \leq 2m, \text{ $i$ odd},\\
        -\frac{\tilde\epsilon_i}{s_2(s_2+\tilde\epsilon_i)},& 1\leq i \leq 2m, \text{ $i$ even},\\
        -\frac{\tilde\epsilon_i}{s_1(s_1+\tilde\epsilon_i)},& 2m+1\leq i \leq 4m, \text{ $i$ even},\\
        -\frac{\tilde\epsilon_i}{s_2(s_2+\tilde\epsilon_i)},& 2m+1\leq i \leq 4m, \text{ $i$ odd}.
    \end{dcases}
\end{align}
The following proposition handles stability of the eigenvalues and is a direct application of the well-known Weyl theorem.
\begin{proposition}\label{thm:eigenvaluestability1}
Let $\hat{\mathcal{C}}$ be the capacitance matrix associated to the structure described in \eqref{eq: perturbation si-s} and let 
\begin{align}\label{equ:errorbound1}
    \epsilon \coloneqq \max_{1\leq i\leq N-2}\vert \epsilon_i\vert+\vert\epsilon_{i+1}\vert.
\end{align}
Then, the eigenvalues $\hat{\lambda}_k$ (sorted increasingly) satisfy
\begin{align}
\label{eq: bound eigenvalues error}
    \vert \hat{\lambda}_k - \lambda_k\vert \leq 2\epsilon, \quad 1\leq k\leq N,
\end{align}
where $\lambda_k$ are the eigenvalues of $\mathcal{C}$.
\end{proposition}
\begin{proof}
    Weyl's theorem states the same result but with the bound of \eqref{eq: bound eigenvalues error} replaced by $\Vert \hat{\mathcal{C}} - \mathcal{C}\Vert$. However, for a tridiagonal matrix $M$ with $\bm 1$ in its kernel, $\Vert M\Vert\leq 2\max_{i}\vert M_{i(i-1)}\vert+\vert M_{i(i+1)}\vert$ by the Gershgorin circle theorem so we obtain the result. 
\end{proof}

Applying \cref{thm:eigenvaluestability1} to our system of dimers where the perturbations $\vert\tilde\epsilon_i\vert\leq \eta$ are in the interval $(-\eta,\eta)$ for some $\eta>0$, we obtain that the eigenvalue perturbation is bounded by 
\begin{align*}
    \frac{2\eta}{s_1(s_1-\eta)}+\frac{2\eta}{s_2(s_2-\eta)}=2\eta\left(\frac{1}{s_1^2}+\frac{1}{s_2^2}\right) + \mathcal{O}(\eta^2)\quad \text{as }\eta\to 0.
\end{align*}

In order to analyse the stability of the eigenvectors, we will use the 
Davis--Kahan theorem \cite{davis-kahan}, which needs some preliminary introduction. Let $E$ and $F$ be $d\times r$ matrices with orthonormal columns such that $\text{span}(E)=\mathcal{E}$ and $\text{span}(F)=\mathcal{F}$. The \emph{canonical angles} between $\mathcal{E}$ and $\mathcal{F}$ are defined as 
\begin{align*}
    \theta_j = \cos\inv \sigma_j \quad 1\leq j\leq r,
\end{align*}
where $\sigma_j$ are the $r$ singular values of $E^TF$. We denote by
\begin{align*}
    \Theta(\mathcal{E}, \mathcal{F}) = \diag(\theta_1,\dots,\theta_r)
\end{align*}
the canonical angles' matrix.

The Davis--Kahan theorem states the following.
\begin{theorem}\label{thm:daviskahanthm} 
Let $S$ and $\tilde S$ be two $d\times d$ symmetric matrices with eigenvalues
\begin{align*}
    \lambda_1\geq \lambda_2\geq\dots\geq \lambda_d,\\
    \tilde\lambda_1\geq \tilde\lambda_2\geq\dots\geq \tilde\lambda_d,
\end{align*}
respectively. Fix $1\leq r\leq s\leq d$ and let $V$ and $\tilde V$ be the matrices having as columns the eigenvectors corresponding to $\lambda_j$ and $\tilde \lambda_j$ for $r\leq j\leq s$. Let $\text{span}(V)=\mathcal{V}$ and $\text{span}(\tilde V)=\tilde{\mathcal{V}}$. Define 
\begin{align*}
    \updelta \coloneqq \inf \{ \vert \lambda - \tilde\lambda\vert : \lambda\in[\lambda_s,\lambda_r], \tilde\lambda\in(-\infty,\tilde\lambda_{s+1}]\cup[\tilde\lambda_{r-1},\infty)\}.
\end{align*}
If $\updelta>0$, then
\begin{align*}
    \Vert \sin\Theta(\mathcal{V}, \tilde{\mathcal{V}})\Vert_2\leq \frac{\Vert S-\tilde S\Vert_2}{\updelta},
\end{align*}
where $\sin\Theta(\mathcal{V}, \tilde{\mathcal{V}})_{ii}=\sin(\Theta(\mathcal{V}, \tilde{\mathcal{V}})_{ii})$. 
\end{theorem}

As a direct consequence of Theorem \ref{thm:daviskahanthm}, we have the following theorem for the stability of the interface eigenmodes. 
\begin{theorem}\label{thm: stability interface modes}
Let $\epsilon< \frac{1}{2}\left(\frac{1}{s_1}-\frac{1}{s_2}\right)$ in (\ref{equ:errorbound1}). Let $\bm v$ and $\hat{\bm v}$ be the eigenvectors corresponding to the eigenvalues $\lambda_i$ and $\hat{\lambda}_i$ in the gap of $\mathcal{C}$ and $\hat{\mathcal{C}}$, respectively. Then
\begin{align}
\label{eq: stabiltiy v gap}
    \Vert \bm v - \hat{\bm v}\Vert_2 &\leq \frac{2\sqrt{2}\epsilon}{\updelta}\\&\leq \frac{2\sqrt{2}\epsilon}{\updelta_0-2\epsilon},\label{eq: stabiltiy v gap apriori}
\end{align}
where $\updelta\coloneqq\min\{\vert\lambda_i-\hat{\lambda}_{i+1}\vert,\vert\lambda_i-\hat{\lambda}_{i-1}\vert\}$ and $\updelta_0 =\min\{\vert\lambda_i-\lambda_{i+1}\vert,\vert\lambda_i-\lambda_{i-1}\vert\}$. The \emph{a priori} estimate \eqref{eq: stabiltiy v gap apriori} holds for $\updelta_0>2\epsilon$.%
\end{theorem}

Remark that for $s_1=1$ and $s_2=2$ and $m$ large we have $\updelta_0\approx0.219$ so that the \emph{a priori} estimate \eqref{eq: stabiltiy v gap apriori} holds for $\epsilon<0.1$. This \emph{a priori} estimate is, however, suboptimal as \cref{fig: stability eve} shows.
\begin{proof}[Proof of \cref{thm: stability interface modes}] The proof immediately follows from \cref{thm:daviskahanthm} with $r=s$ by using the bound $\Vert \bm v - \hat{\bm v}\Vert_2\leq\sqrt{2}\sin\Theta(\bm v,\hat{\bm v})$ \cite{DK_Stats}. 
\end{proof}
In \cref{fig: stability} we show numerically the high stability of the interface modes. We consider perturbations of the type $\tilde\epsilon_i\sim U(-\eta,\eta)$ where we call $\eta$ the perturbation size, we display the latter as percentage relative to the resonator's size. \cref{fig: stability eva} shows that the interface eigenfrequency (lying in the gap) is only minimally perturbed even by perturbation in the size of $20\%$. We remark that numerically the bound of \eqref{eq: bound eigenvalues error} can even be sharpened to $\vert \hat{\lambda_k} - \lambda_k\vert \leq \frac{3}{2}\epsilon$. \cref{fig: stability eve} shows $\left\Vert\bm v - \hat{\bm v}\right\Vert_2$ for various perturbation sizes and normalised $\bm v$ and $\hat{\bm v}$. The black lines shows the average over $10^4$ runs while the gray area encloses the range from the minimum to the maximum value of $\left\Vert\bm v - \hat{\bm v}\right\Vert_2$.
\begin{figure}[h]
    \centering
    \begin{subfigure}[t]{0.48\textwidth}
    \centering
    \includegraphics[height=0.76\textwidth]{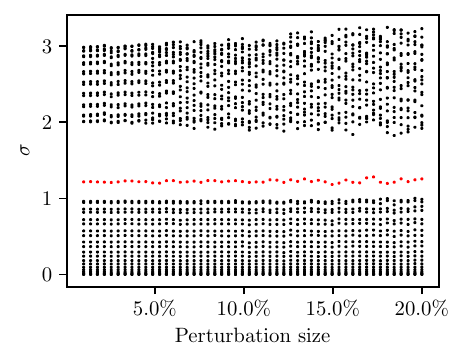}
    \caption{Spectrum of the capacitance matrix with perturbations given by $\tilde\epsilon_i\sim U(-\eta,\eta)$. For every perturbation size the spectrum of one realisation is shown. The eigenvalue in red corresponds to the localised interface mode.}
    \label{fig: stability eva}
    \end{subfigure}\hfill
    \begin{subfigure}[t]{0.48\textwidth}
    \centering
    \includegraphics[height=0.76\textwidth]{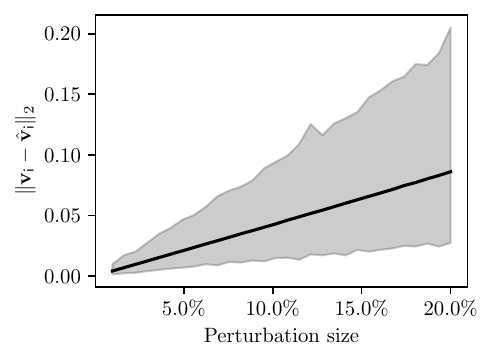}
    \caption{Stability of the interface mode. The solid black line shows the average dislocation over $10^4$ runs, while the gray area encloses the range from the minimum to the maximum dislocation observed over these realisations.}
    \label{fig: stability eve}
    \end{subfigure}
    
    \caption{The interface eigenvalue and the corresponding interface mode are very stable also in presence of big perturbations. Simulations in a system of $N=41$ resonators with $s_1=1$ and $s_2=2$. Perturbations are uniformly distributed in $(-\eta,\eta)$ where we call $\eta$ the perturbation size, expressed relatively to the resonators' sizes.}
    \label{fig: stability}
\end{figure}
\section{Concluding remarks}  \label{sect7}
In this paper, we have quantitatively characterised interface eigenmodes in finite, dimer systems of subwavelength resonators with a geometric defect and proved the exponential decay of their associated eigenmodes. Our characterisation is based on (broken) translation invariance properties of the associated capacitance matrix together with  properties of Chebyshev polynomials. The  $3$-term recurrence relation satisfied by the  Chebyshev polynomials is shown to be useful for analysing spectra of tridiagonal (perturbed) $2$-Toeplitz matrices. 

Following this line of research, it would be very interesting to generalise the results obtained in this paper to extended (known also as multi-band) SSH models, exhibiting  exponentially localised interface eigenmodes with corresponding eigenfrequencies inside the multiple subwavelength band gaps of the structure \cite{extendedSSH1,extendedSSH2}. Another highly interesting direction would be to extend the current results to finite dimer systems of three-dimensional subwavelength resonators. In the case of three-dimensional systems of subwavelength resonators, the main difficulty  occurs from the long-range interactions between the resonators, which lead to a slow decay of the off-diagonal terms of the corresponding capacitance matrix. Nevertheless, in view of the recent results on  $k$-banded approximations of three-dimensional capacitance matrices \cite{skin3d} together with the convergence spectral results as the size of the structure goes to infinity in \cite{ammari2023perturbed,essentiel2023},  it may be possible to prove existence and uniqueness of localised eigenmodes in finite chains of subwavelength resonators in three dimensions, as numerically shown in \cite{ssh3d}.  
 Another very interesting problem is to relate the localisation effect at the interface to the statistics of the eigenvalues of the capacitance matrix under random perturbations in the parameters of the system. This would extend the well-known Thouless localisation criterion \cite{thouless} to interface eigenmodes. 

\vspace{1cm}
\noindent
\textbf{Data Availability}

The data that support the findings of this study are openly available at \url{https://zenodo.org/doi/10.5281/zenodo.10361315}.

\vspace{0.5cm}
\noindent
\textbf{Conflict of interest} 

The authors have no competing interests to declare that are relevant to the content of this
article.

\vspace{0.5cm}
\noindent
\textbf{Acknowledgments}

This work was supported by Swiss National Science Foundation grant number 200021--200307 and by the Engineering and Physical Sciences
Research Council (EPSRC) under grant number EP/X027422/1.

\appendix

\section{Pseudo-spectrum of a normal matrix} \label{appA}
\begin{definition}
$\sigma_{\varepsilon}(\mathbf{A})$ is the set of $z \in \mathbb{C}$ such that
$$
\|(z-\mathbf{A}) \mathbf{v}\|<\varepsilon,
$$
for some $\mathbf{v} \in \mathbb{C}^N$ with $\|\mathbf{v}\|=1$.
\end{definition}

The next theorem expresses these facts with the aid of the following notation for an open $\varepsilon$-ball:
$$
\Delta_{\varepsilon}=\{z \in \mathbb{C}:|z|<\varepsilon\} .
$$
In this theorem, a sum of sets has the usual meaning:
$$
\sigma(\mathbf{A})+\Delta_{\varepsilon}=\left\{z: z=z_1+z_2, z_1 \in \sigma(\mathbf{A}), z_2 \in \Delta_{\varepsilon}\right\},
$$
which is equal to $\{z: \operatorname{dist}(z, \sigma(\mathbf{A}))<\varepsilon\}$.

\begin{theorem}\label{thm:pseudospectraofnormaloperator}[Pseudo-spectrum of a normal matrix] 
For any $\mathbf{A} \in \mathbb{C}^{N \times N}$,
$$
\sigma_{\varepsilon}(\mathbf{A}) \supseteq \sigma(\mathbf{A})+\Delta_{\varepsilon} \quad \forall \varepsilon>0,
$$
and if $\mathbf{A}$ is normal and $\|\cdot\|=\|\cdot\|_2$, then
\begin{equation}\label{equ:pseudospectraequ1}
\sigma_{\varepsilon}(\mathbf{A})=\sigma(\mathbf{A})+\Delta_{\varepsilon} \quad \forall \varepsilon>0 .
\end{equation}
Conversely, if $\|\cdot\|=\|\cdot\|_2$, then (\ref{equ:pseudospectraequ1}) implies that $\mathbf{A}$ is normal.
\end{theorem}

\section{Topological origin} \label{appC} 
Infinite SSH structures have long been known for their topological nature. In this section, we show that a topological invariant can be defined also for finite structures in such a way that when the size of the system grows they converge to the infinite invariants. For one-dimensional crystals, the \emph{Zak phase} has been shown to be related through an \emph{if and only if} condition to the existence of interface modes, see \emph{e.g.} \cite[Theorem 1]{coutant2023surface} and also \cite{fefferman1, hai2022}. Denoting by $u_j^\alpha$ a family of eigenmodes depending piecewise smoothly on the quasiperiodicity $\alpha$, the Zak phase is defined as 
\begin{align}
    \label{eq: def Zak}
    \phi^{\text{zak}}_j \coloneqq \mathrm{i}\int_{Y^*} \left\langle u_j^\alpha,\frac{\partial}{\partial \alpha} u^\alpha_j\right\rangle \dd \alpha,
\end{align}
where $Y^*$ denotes the first Brillouin zone and $\langle \cdot, \cdot \rangle$ the usual $L^2$ inner product. The Zak phase is also related to the symmetries of the eigenfunctions \cite{coutant2023surface}. Specifically, denoting by $u^+_j$ the eigenfunction associated to the quasifrequency that maximises the  $j$-th band function $\omega_j^\alpha$ we may define bulk topological index of the $j$-th band gap as
\begin{align}
    \label{eq: def indicator}
    \mathcal{J}_j \coloneqq \begin{dcases}
        +1, & u^+_j(x) = \mathcal{P}(u^+_j),\\
        -1, & u^+_j(x) = -\mathcal{P}(u^+_j),
    \end{dcases}
\end{align}
for $j=1,2,$ where $\mathcal{P}$ denotes the mirroring of the unit cell $Y$ about its center. Then, the Zak phase is related to $\mathcal{J}_j$ via
\begin{align}
    \label{eq: indicator and zak}
    \mathcal{J}_j = (-1)^{j-1}\prod_{k=1}^{j}e^{\mathrm{i}\phi^{\text{zak}}_k}.
\end{align}

Previous work \cite[Proposition 5.5]{ammari.ea2023Edge} has shown that the Zak phase of a periodic system of dimers is quantised and depends on the inter- and intra-spacing between the cells. In particular,
\begin{align*}
    \phi^{\text{zak}}_j = \begin{dcases}
        \pi, &s_1 \geq s_2,\\
        0, & s_1 < s_2,
    \end{dcases}
\end{align*}
and thus 
\begin{align*}
    \mathcal{J}_1 = \begin{dcases}
        -1, &s_1 \geq s_2,\\
        1, & s_1 < s_2.
    \end{dcases}
\end{align*}
We will show now that \eqref{eq: def indicator} lends itself well to a reformulation in the finite case. We denote by $\mathcal{P}^{\text{d}}$ the discrete equivalent of $\mathcal{P}$ that is, for $\bm v \in \R^{2k}$,
\begin{align*}
    \mathcal{P}^{\text{d}}(\bm v)^{(j)} \coloneqq \begin{dcases}
        \bm v^{(j-1)},& 1\leq j\leq 2k,\text{ $j$ even},\\
        \bm v^{(j+1)},& 1\leq j\leq 2k,\text{ $j$ odd},
    \end{dcases}
\end{align*}
and define the discrete equivalent of \eqref{eq: def indicator} by
\begin{align}
    \label{eq: def indicator discrete}
    \mathcal{J}^{\text{d}}(\bm v) \coloneqq \frac{1}{\Vert u\Vert_2^2} \left\langle \bm v,\mathcal{P}^{\text{d}}(\bm v)\right\rangle.
\end{align}
In \cref{fig: discrete indicator} we show the values of $\mathcal{J}^{\text{d}}(\bm v)$ for two structures composed of $40$ dimers. \cref{fig: indicator12} shows the case $s_1=1$ and $s_2=2$ while \cref{fig: indicator21} shows the case $s_1=2$ and $s_2=1$. We observe a very different behaviour, but in particular the value of interest in view of \eqref{eq: def indicator} is the one at $\lambda=1$, where the two structures take values $\approx+1$ and $\approx-1$ respectively showing that through \eqref{eq: def indicator discrete} we may generalise \eqref{eq: def indicator} to discrete structures.
    \begin{figure}[h]
    \centering
    \begin{subfigure}[t]{0.48\textwidth}
    \centering
    \includegraphics[height=0.76\textwidth]{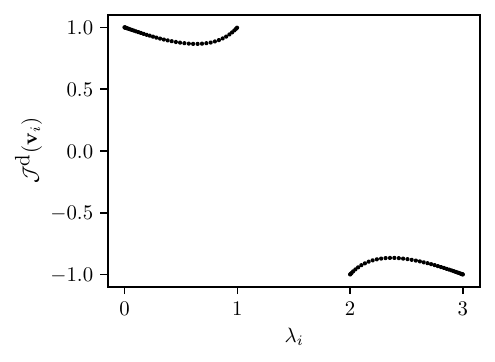}
    \caption{Discrete indicator $\mathcal{J}^{\text{d}}(\bm v)$ $s_1=1$ and $s_2=2$.}
    \label{fig: indicator12}
    \end{subfigure}\hfill
    \begin{subfigure}[t]{0.48\textwidth}
    \centering
    \includegraphics[height=0.76\textwidth]{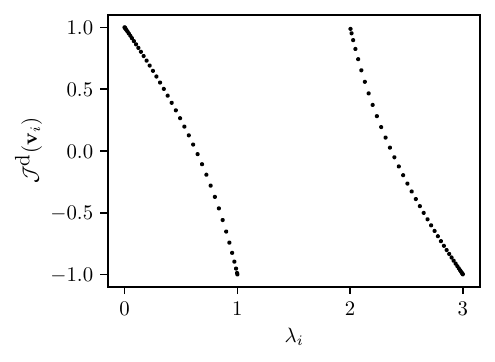}
    \caption{Discrete indicator $\mathcal{J}^{\text{d}}(\bm v)$ $s_1=2$ and $s_2=1$.}
    \label{fig: indicator21}
    \end{subfigure}
    \caption{The discrete indicator function of \eqref{eq: def indicator discrete} have very different behaviour for $s_1<s_2$ and $s_1>s_2$. The values at the end of the bands approximate accurately the values of the infinite periodic case of \eqref{eq: def indicator}.}
    \label{fig: discrete indicator}
\end{figure}

\printbibliography

\end{document}